\documentclass[aip,jmp,rsi,amsmath,amssymb,reprint,amgraphicx]{revtex4-1}

\usepackage{graphicx}
\usepackage{dcolumn}
\usepackage{bm}

\usepackage{amsmath,amsfonts,amssymb,amsthm}

\newcommand{\dfin}[2]{#1\in\mathcal{A}^{#2}}

\newcommand{\hodgeop}{\star}

\newcommand{\Laplace}{\triangle}
\newcommand{\poincare}{Poincar\'{e} }

\theoremstyle{plain}
\newtheorem{proposition}{Proposition}
\newtheorem{theorem}{Theorem}
\newtheorem{definition}{Definition}
\theoremstyle{remark}
\newtheorem{remark}{Remark}

\draft 

\begin{document}

\title{\poincare Duality in the Quantization of the Norm of Differential Forms} 

\author{Juan Mendez}
\email{juanmenrod58@gmail.com}
\affiliation{Former Professor. Universidad de Las Palmas de Gran Canaria, Canary Island. Spain}

\date{\today}

\begin{abstract}
The more important difference between Riemann and pseudo-Riemann manifolds is the metric signature and its theoretical consequences. The practical application for Physics Theories becomes often impossible due to the signature consequences.
Eg., some of the rich results in Riemann Geometry and Topology become invalid for Physics if they are based on the concept of the positive definite norm; to avoid this problem, the proof machinery must avoid such assumption and must be based in other tools. This paper is a contribution to provide methodologies for Hodge decomposition and \poincare duality based on the concept of linear independence of canonical classes instead of the positive norm. 

As a result, the Hodge and norm decompositions are expressed based on continuous and discrete terms.  When this result is applied to Classical Electromagnetic Theory, in pseudo-Riemann manifolds with minkowskian metric, magnitudes as the field norm and action have one discrete sum of terms. This result, as a quantization of the norm and action is a property of the Topology, in special of the Cohomology classes, that are sources of the field as well as the generators of action quantum. 
\end{abstract}

\pacs{02.40.-k; 02.40.Vh; 03.50.De; 14.80.Hv}

\keywords{Pseudo-Riemann manifolds; Non-simple connected canifolds;  Cohomology classes; Hodge decomposition; \poincare duality; Electromagnetic duality; Electromagnetic action.}

\maketitle

\section{Introduction}\label{sec:introduction}
The Hodge decomposition of any $p$-form $\phi$  in  Riemann manifolds allows a canonical decomposition as\cite{Flanders:2012,Morita:2001}:
\begin{equation}\label{eq:hodgedecomposition1}
\phi = d\alpha + \delta\beta + \phi_{h}
\end{equation}

where $d$ and $\delta$ are the exterior derivative and coderivative\cite{Gockeler:1989} respectively, $\alpha$ is a $(p-1)$-form, $\beta$ is a $(p+1)$-form and $\phi_{h}$ is an harmonic form characteristic of the cohomology classes. 
The decomposition establishes  that there is solution for $\alpha$, $\beta$ and $\phi_{h}$. 

Some proofs of the Hodge decomposition, eg. Morita\cite{Morita:2001}
and Jost~\cite{Jost:2011}, are based on that in Riemann manifolds the norm of differential forms is positive definite. This property is central in the theory of harmonic forms, but this concept can be different between Riemann and pseudo-Riemann manifolds. 

This paper provides an alternative method to study the formal properties of the Hodge Decomposition. Our approach is based on the linear independence of canonical differential forms instead of being based on the use of the positive norm. In this process, \poincare duality\cite{Hatcher:2002,Bott:1982} plays a remarkable role, mainly in the decomposition of the norm including a sum of discrete terms based on the integral of cohomology representative forms. 
The norm value becomes the sum of some discrete values, as a quantization, where each cohomology class contributes with one discrete value. 

\poincare and Hodge duality are some type of \emph{complement} relationship to complete the unit or volume form; that requires that all the forms are quadratically integrable. However, beyond that resemblance, both dualities are applied to different theoretical objects. 

This property of the quantized norm of differential forms is applied to physical theories. In the case of Electromagnetism in no simple connected manifolds, the Electromagnetic Action becomes quantized such as each cohomology class provides one discrete action value or quantum.    

Although the Wormholes approaches\cite{Misner:1957} are also based in Topology, the approach of this paper uses the concepts abstractly, neither assumption is included about how the space-time is nor what it is its metric. However, the described picture has in common with Wormholes theories the concept of field without needing elementary particles\cite{Misner:1957}. That is the Field as a property of the space-time. 

The plan of this paper is as follows, Section \ref{sec:framework} presents a summary of theoretical materials used as the cohomology classes and the \poincare duality. Section \ref{sec:Hodge} presents the Hodge decomposition from the viewpoint of linear independent classes and the \poincare duality. Section \ref{sec:norm} presents the decomposition of the norm that includes some discrete, or quantized, terms associated with each cohomology class. The special case of even-dimensional manifolds is presented in Section \ref{sec:even}. An application for generalized dual Electromagnetism is presented in Section  \ref{sec:aplicacion}. Finally, a Conclusions Section is included.

\section{Theoretical Framework}\label{sec:framework}
Let $(M,g)$ be a $n$-dimensional compact, differentiable, oriented and connected manifold $M$ with a metric $g$ locally reducible to a diagonal case:
\begin{equation}
\eta=\textrm{diag}(\;\underbrace{1,\ldots,1}_{r},\underbrace{-1,\ldots,-1}_{s}\;)
\end{equation}
where $r$ and $s=n-r$ are the number of positive and negative ones respectively. If both $r$ and $s$ are non null, it  is named a pseudo-Riemann, or semi-Riemann, manifold with indefinite metric, while  Riemann manifold, $s=0$, is a case that has positive definite metric.

Although there are many similarities between Riemann and pseudo-Riemann manifolds\cite
{Bishop:1968}, there are some differences in the properties that are consequences of the metric signature;  in differential forms mainly in the computation of the Hodge dual and its subsequent uses. Many properties depend on this duality. Therefore,  many results in the extensive bibliography in Riemann manifolds must be carefully used when applied to the pseudo-Riemann case.     

Let $\mathcal{A}^{p}(M,\mathbb{R})$ be the set of $p$-forms on $M$ with values in $\mathbb{R}$. 
The Hodge duality gets an isomorphism between $\mathcal{A}^{p}$ and $\mathcal{A}^{n-p}$. The Hodge star operator, $\hodgeop$, defines a linear map $\hodgeop : \mathcal{A}^{p} \rightarrow \mathcal{A}^{n-p}$. The Hodge duality is closely related to the unit $n$-form $\omega$, such as is $\alpha$ is a $p$-form and $\hodgeop\alpha$ is its dual. It is verified:
\begin{equation}
\alpha\wedge\hodgeop\alpha= <\alpha,\alpha>\omega \qquad \int_{M}\omega=1
\end{equation}
where $<\alpha,\alpha>$ is their inner product or square of their norm dependent on the metric signature, eg. for a vector $V$, it is: $<V,V>=V_{a}V_{b}g^{ab}$. The Hodge dual of the unit $n$-form $\omega$ is the scalar or $0$-form $1$; the relationships between both are: $\hodgeop 1=\omega$ and $\hodgeop\omega = (-1)^{s}$. The double Hodge duality for $\dfin{\phi}{p}$\cite{Gockeler:1989} verifies:
\begin{equation}
\hodgeop\hodgeop  \phi = (-1)^{D(p)}\phi \qquad D(p)=p(n-p)+s
\end{equation}
It is verified that: $D(p)=D(n-p)$. The exterior derivative, $d$, defines a linear map $d:\mathcal{A}^{p}\rightarrow\mathcal{A}^{p+1}$. It  allows the definition of the coderivative  $\delta:\mathcal{A}^{p}\rightarrow\mathcal{A}^{p-1}$ as\cite{Gockeler:1989}:
\begin{equation}
\delta\phi = (-1)^{C(p)}\hodgeop  d\hodgeop  \phi \qquad C(p)=np+n+1+s
\end{equation}
The Hodge duality allows to define a bilinear integral in the manifold $M$ of two $p$-forms $\alpha$ and $\beta$ as:
\begin{equation}\label{eq:norm}
(\alpha,\beta) = \int_{M} \alpha\wedge\hodgeop\beta
\end{equation}
The bilinear integral allows the definition of the norm of a differential form  as: $(\alpha,\alpha)$. In Riemann manifolds, it verifies: $(\alpha,\alpha)\geq 0$. If it is verified that $(\alpha,\alpha)= 0$, then must be $\alpha=0$; so it is positive definite for all non null forms\cite{Jost:2011}.

Derivative and coderivative operators have  some similar properties eg.  derivative verifies that: $dd=0$ and coderivative verifies: $\delta\delta = 0$.  It is verified that $\alpha\wedge\hodgeop\beta=\beta\wedge\hodgeop\alpha$; from this the bilinear integral in Equation (\ref{eq:norm}) has some properties~\cite{Morita:2001}
common for Riemann and pseudo Riemann manifolds as: $(A,B)= (B,A)$. Also, other properties are: $(dC,A)=(C,\delta A)$, where $A$ and $B$ are $p$-form and $C$ is a $(p-1)$-form, also the linear property: $(A,B+C)=(A,B)+(A,C)$. 

\subsection{Strong harmonic forms}\label{subsec:harmonic}
Let $\Laplace$ be a  second order differential operator --called Laplace-Beltrami--  that map $\Laplace : \mathcal{A}^{p}\rightarrow \mathcal{A}^{p}$. It is defined as: $ \Laplace = \delta d + d \delta$. A $p$-form $\phi$ is named harmonic if it verifies: $\Laplace \phi=0$. In Riemann manifolds this implies that $\phi$ is closed, $d\phi=0$, and also its dual is closed, that is $d\hodgeop\phi=0$. That last condition can be called co-closed or dual closed (its dual is closed) that implies that: $\delta\phi=0$. If both conditions are verified, it is an harmonic form, also called by Bott and Tu\cite
{Bott:1982} global closed.

However, in pseudo-Riemann manifolds, there are some differences concerning these concepts. In this paper, a differential form $\phi$ verifying $d\phi=0$ and also $\delta\phi=0$ is named strong harmonic to indicate that it is included in the harmonic class but the reciprocal is not true in pseudo-Riemann manifolds.

With a positive definite metric, the equation $\Laplace\phi=0$ defines an elliptic second-order differential equation. However, with an indefinite metric, it defines a hyperbolic second order differential equation, whose general solution is a wave. 

Instead of a theoretical proof, real-world evidence is more illustrative; this is the existence of Electromagnetic waves in the vacuum. 
These can be represented by a harmonic $1$-form, $\Laplace A=0$, also a gauge condition can be imposed $\delta A=0$, but the Electromagnetic Field, $F=dA$,  is no null. Therefore, Electromagnetic waves in the vacuum are  physical examples of harmonic but no strong harmonic forms. The reason is rather simple because the space-time  --the theoretical framework of physical phenomena-- is modeled as a pseudo-Riemann manifold with Minkowski-like metric, instead of a \emph{pure} Riemannian one. The conclusion is that not all results and methodologies coming from Riemannian manifold literature can be directly transferred to the pseudo-Riemannian case, mainly if harmonic forms are involved. 
    
In this paper the existence of solution for the equation: $\Laplace  \theta = \varphi$ is supposed. This has been deeply analysed in hyperbolic differential equations\cite{Friedlander:1975}. The solution can be represented by: $\theta = G\circ \varphi$, where the $G=\Laplace^{-1}$ operator is the Green map $G: \mathcal{A}^{p}\rightarrow\mathcal{A}^{p}$. 

\subsection{ Cohomology Classes}\label{subsec:cohomology}
Let $Z^{p}(M)$ be  the set of all closed $p$-forms in $M$ and $B^{p}(M)$ the set of all the exact $p$-forms. The quotient set $H^{p}(M)= Z^{p}(M)\setminus B^{p}(M)$ contains all the closed but non-exact $p$-forms in $M$. The dimensionality of the cohomology $H^{p}(M)$, which is finite dimensional, is the $p$-dimensional Betti number 

A $p$-rank cycle $z$ on $M$ is a closed $p$-dimensional sub-manifold on $M$; that is without boundary, $\partial z =\emptyset$. A simply connected manifolds $M$ has not non null cycles,  $\textrm{dim} H^{p}(M)=0$, then $Z^{p}(M)$ and $B^{p}(M)$ are equivalent and all the closed forms are exact ones, this is the \poincare Lemma\cite{Gockeler:1989}. But in no simple connected manifolds not all closed forms ($d\omega=0$) are exacts ($\omega \neq d\phi$). This difference is the theoretical arena  of the cohomology classes.  


\subsubsection{The de Rham Theorem}\label{subsubsec:FirstRham}
According to the  de Rham Theorem\cite{Flanders:2012}
the necessary and sufficient condition for a form $\phi \in \mathcal{A}^{p}(M)$ be exact is that $\int_{z}\phi=0$ for any $p$-dimensional cycle $z$ in $M$. Let $[\phi]$ be a representative form of the cohomology class to which $\phi$ belongs, it is verified that if $\phi = [\phi]+ d\varphi$ them both $\phi$ and $[\phi]$ are cohomologous forms. Any cohomologous form can be chosen as the representative of the class. The integral in a cycle $z$ of any cohomologous form is a characteristic of the class:
\begin{equation}
\int_{z} \phi = \int_{z}[\phi] + \int_{z}d\varphi = \int_{z}[\phi]
\end{equation}
because by  applying the Stokes theorem in the closed submanifold $z$, it is verified that: $\int_{z}d\varphi=\int_{\partial z}\varphi=0$.

The Hodge Theorem, as shown by Jost\cite
{Jost:2011} and Morita~\cite
{Morita:2001}, proof that every cohomology class $H^{p}(M)$ contains precisely one harmonic form. Hence, as asserted by Morita\cite
{Morita:2001}, let $z^{(p)}_{1},\ldots,z^{(p)}_{\beta_{p}}$ be a set of linearly independent $p$-dimensional cycles. Each cycle can be characterized by one homology class; therefore one strong harmonic form $\gamma^{(p)}_{a}\in\mathcal{A}^{p}(M,\mathbb{R})$ is chose as the representative of the $a^{\textrm{th}}$ class in $H^{p}(M)$. 

\begin{definition}
Let $\gamma^{(p)}(M)=\{\gamma^{(p)}_{1},\ldots,\gamma^{(p)}_{\beta_p}\}$ be a set of of linear independent $p$-forms representative of the classes in  $H^{p}(M)$.  To each cycle corresponds one representative form and vice versa: $z^{(p)}_{a} \leftrightarrow \gamma^{(p)}_{a}$.
\end{definition}

According with the de Rham Theorem\cite{Flanders:2012,Morita:2001}, there is a set of real numbers: $w_{1},\ldots,w_{r}$ for which there is an unique closed $p$-form $\phi$ which verifies: $\int_{z^{p}_{a}}\phi = w_{a}$. This form is undefined only in an arbitrary exact form. The following closed $p$-form $\phi$ must verify the assertion of the de Rham Theorem: 
\begin{equation}\label{nonormal}
\phi = \sum_{a=1}^{\beta_{p}}w_{a}\gamma^{(p)}_{a}
\end{equation}
The integral of a cohomologous form in a cycle must be non null, i.e. $\int_{z}\gamma \ne 0$, if and only if $\gamma$ is the cohomologous correspond of $z$. Linear independence of representative forms means that if: $\sum_{a=1}^{\beta_{p}}w_{a}\gamma^{(p)}_{a} = 0$, then it must be:  $w_{a}=0$. The normalizations for the set of representative forms is:
\begin{equation}\label{eq:normalization1}
\int_{z^{p}_{b}} \gamma^{(p)}_{a} = \delta_{ab}  \qquad \int_{M} \gamma^{(p)}_{a} \wedge \hodgeop \gamma^{(p)}_{a} = \lambda^{(p)}_{aa} 
\end{equation}
another significant integral is the extension of the norm: $(\gamma^{(p)}_{a},\gamma^{(p)}_{b})= \int_{M} \gamma^{(p)}_{a}\wedge\hodgeop\gamma^{(p)}_{b}=\lambda^{(p)}_{ab}$.  The set $\gamma(M)$ represents the global properties of the manifold topology, being an intrinsic property that characterize it. It is as basic as the metric $g$ that represents the differential geometry or local properties of the manifold. Both the metric and the characteristic forms $(M;g,\gamma)$ must be used to explicit the role of both, the local/differential and global/topological properties of the manifold.

\subsection{The \poincare duality}\label{subsec:PoincareDuality}

The Hodge duality does a map $\mathcal{A}^{p}\rightarrow \mathcal{A}^{n-p}$ for every $\dfin{\phi}{p}$ while \poincare duality\cite{Hatcher:2002,Bott:1982}  does a map for the chomology classes as: $H^{p}\rightarrow H^{n-p}$. Simplified, \poincare duality defines a map for each pair of cohomologous forms in $H^{p}(M)$ and $H^{n-p}(M) $ into $\mathbb{R}$, that is: 
\begin{equation}
H^{p}(M) \times H^{n-p}(M) \rightarrow \mathbb{R} 
\end{equation}

From this Equation, the \poincare Duality Theorem~\cite
{Morita:2001} establishes that for any non null cohomology class $[\alpha]\in H^{p}(M)$ exists a no null cohomology class $[\beta]\in H^{n-p}(M)$ such as $\int_{M} \alpha\wedge\beta \neq 0$. This duality is an isomorphism, $H^{p}\cong H^{n-p}$, therefore exist one and only one class $[\beta]\in H^{n-p}(M)$ for each $[\alpha]\in H^{p}(M)$ and the reciprocal. Due to the isomorphism between the sets $\boldsymbol{\gamma}^{(p)}$ and $\boldsymbol{\gamma}^{(n-p)}$, their dimensionalities or Betti numbers have the same value: $\beta_{p} = \beta_{n-p}$. 

\begin{definition}\label{defmatrixE}
The matrix: $\mathbf{E}^{(p)}=(\varepsilon^{(p)}_{ab}) \in M(\beta_{p}, \mathbb{R})$, defined as follows, has one and only one non null element for each row and column, therefore it is no singular.
\begin{equation}\label{eq:firstdefinition}
\varepsilon^{(p)}_{ab}=\int_{M}\gamma^{(p)}_{a}\wedge\gamma^{(n-p)}_{b} 
\end{equation}
\end{definition}

An alternative definition of \poincare duality is presented by Bott and Tu\cite{Bott:1982}; it focuses in the relationship between $z$ cycle and representative form,  such as if $z$ is a closed oriented sub-manifold of dimension $p$ (a $p$-cycle) in an oriented manifold $M$ of dimension $n$, the \poincare dual of $z$ is the cohomology class of closed $(n-p)$-form $\eta_{z}\in H^{n-p}(M)$ characterized by the following property.
\begin{equation}\label{eq:bottdefinition}
\int_{z} \phi = \int_{M} \phi\wedge\eta_{z}
\end{equation} 

Using the previous definitions and theoretical materials, the Bott and Tu definition can be rewritten based on $\gamma^{(p)}$ set. Based on Equation \eqref{eq:bottdefinition}, it is verified that:  
\begin{equation}\label{eq:bottdefinition2}
\int_{z^{(p)}_{a}} \gamma^{(p)}_{b} = \int_{M} \gamma^{(p)}_{b}\wedge\eta^{(n-p)}_{a} =\delta_{ab}
\end{equation} 

where $\eta^{(n-p)}_{a}$ is the \poincare dual of $z^{(p)}_{a}$. However, $\eta^{(n-p)}_{a}\in H^{n-p}(M)$ is not necessarily one of the classes $\gamma^{(n-p)}_{a}$ previously presented. A linear combination can be used: 
\begin{equation}
\eta^{(n-p)}_{a} = \sum_{c=1}^{\beta_{n-p}} \eta_{ac} \gamma^{(n-p)}_{c} \qquad \eta_{ac}  \in \mathbb{R}
\end{equation} 

This implies that:
\begin{equation}\label{eq:normalfactor}
\sum_{c=1}^{\beta_{n-p}} \eta_{ac} \int_{M} \gamma^{(p)}_{b}\wedge\gamma^{(n-p)}_{c} = \delta_{ab}
\end{equation} 

It can be rewrite based on the matrix $\varepsilon^{(p)}_{ab}$ as: $\sum_{c=1}^{\beta_{n-p}} \eta_{ac} \varepsilon^{(p)}_{bc}=\delta_{ab}$. If one and only one of the row and column values of $\varepsilon^{(p)}_{ab}$ is non null, then one and only only one of the values of $\eta_{ab}$ can be chose as non null. If $\varepsilon^{(p)}_{ab}\neq 0$, then $\eta_{ab} = 1/\varepsilon^{(p)}_{ab}$. That means that the difference between the dual of $z^{(p)}_{a}$, in the Bott-Tu definition way, and the dual of $\gamma^{(p)}_{a}$ is a multiplicative constant.

To simplify the expressions dealing \poincare duality, we will use the symbol $P$ in two ways: as operator: $P: H^{k}(M)\rightarrow H^{n-k}(M)$; also as an index function/modifier, such as we will use $b=P(a)$ to express that the index $b$ corresponds to the \poincare dual of index $a$. This means that:
\begin{equation}\label{eq: poincareoperador}
P\gamma^{(p)}_{a} = \gamma^{(n-p)}_{P(a)}
\end{equation}

If $a$ is a index in $\gamma^{(p)}$ set, then $P(a)$ is a index in set $\gamma^{(n-p)}$ and vice versa. Due to the isomorphism, must be that $P(P(a))=a$. Also, as is shown in Figure \ref{fig:classes}, it is verified that: 
\begin{equation}\label{eq: poincareoperadord}
PP\gamma^{(p)}_{a} = P\gamma^{(n-p)}_{P(a)} = \gamma^{(p)}_{P(P(a))} = \gamma^{(p)}_{a}
\end{equation}

By using this index definition, the non null value in Equation (\ref{eq:firstdefinition}) is that involves the $a$ and $P(a)$ indexes: 
\begin{equation}
\varepsilon^{(p)}_{a,P(a)}=\int_{M}\gamma^{(p)}_{a}\wedge\gamma^{(n-p)}_{P(a)} 
\end{equation}

Bott and Tu\cite{Bott:1982} definition, in Equation (\ref{eq:bottdefinition}), becomes:
\begin{equation}\label{eq:bottdefinition3}
\int_{z^{(p)}_{a}} \phi = \eta\int_{M} \phi\wedge\gamma^{(n-p)}_{P(a)} \qquad 
\eta=\left(\varepsilon^{(p)}_{a,P(a)}\right)^{-1}
\end{equation} 

\begin{figure}[t]
	\centering
	\includegraphics[width=0.3\textwidth]{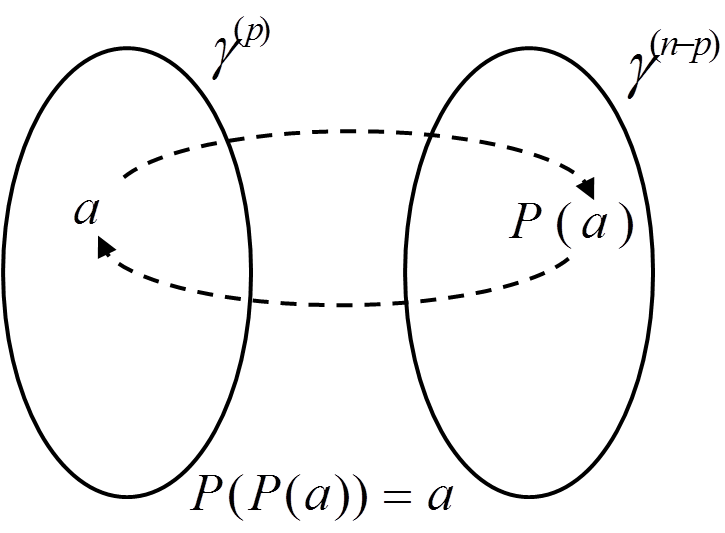}
	\caption{Indexes correspondence of characteristic forms in the isomorphism between $H^{(p)}$ and $H^{(n-p)}$.} \label{fig:classes}	
\end{figure}

\section{The \poincare Duality in the Hodge Decomposition}\label{sec:Hodge}

\begin{definition}
	Let $Z^{p}(M)$  be the class of closed form $p$-forms in $M$, $\hodgeop Z^{n-p}(M)$ the class of dual closed $p$-forms and $W^{p}(M)$ the class of strong harmonic forms, being: $W^{p}(M)= Z^{p}(M) \cap\hodgeop Z^{n-p}(M)$
\end{definition}

\begin{proposition}\label{prop:linearindependence}
	The three classes $Z^{p}\setminus W^{p}$, $\hodgeop Z^{n-p}\setminus W^{p}$ and $W^{p}$ are linearly independent.
\end{proposition}

\begin{proof}
	If $\phi\in Z^{p}\setminus W^{p}$ ($d\phi=0, \delta\phi\neq 0$), $\psi\in \hodgeop Z^{n-p}\setminus W^{p}$ ($d\psi \neq 0, \delta\psi = 0$) and  $\varphi\in W^{p}$ ($d\varphi=0, \delta\varphi = 0$) are non null $p$-forms.
	It must be proved that in the equation $A\phi + B\psi + C\varphi=0$, with constant $\{A, B,C\}\in \mathbb{R}$, the only solution must be: $A=B=C=0$. By using the derivation: $d(A\phi + B\psi +C\varphi)= Bd\psi=0$, therefore it must be $B=0$. By using the coderivative: $\delta(A\phi + B\psi + C\varphi)= A\delta\phi=0$, therefore it must be $A=0$ which also implies that $C=0$.
\end{proof}

\begin{definition}
	Let $\mathcal{H}^{p} \subseteq W^{p}$ be a $p$-form class based on linear combinations of the $\gamma^{p}$ set of the $p$-forms representative of the cohomology classes of $H^{p}(M)$ and $w_{a}\in\mathbb{R}$: 
	\begin{equation}
	\sum_{a=1}^{\beta_{p}}w_{a}\gamma^{(p)}_{a} \in \mathcal{H}^{p} \subseteq W^{p} 
	\end{equation}
\end{definition}

\begin{definition}
	Let $d\mathcal{A}^{p-1}\equiv B^{p}(M)$ be the class of exact $p$-forms, but not dual exact ones;  if $\sigma\in d\mathcal{A}^{p-1} \subset \mathcal{A}^{p}$, it is verified that:  $d\sigma=0$ and $\delta\sigma\neq 0$. It can be expressed as: $\sigma = d\phi$ where $\phi\in \mathcal{A}^{p-1}$.
\end{definition}

\begin{definition}
	Let $\delta\mathcal{A}^{p+1}$ be the class of dual exact $p$-forms, but not exact ones, 
	that is $\sigma = \delta\phi$ where $\sigma\in \delta\mathcal{A}^{p+1} \subset \mathcal{A}^{p}$ and $\phi\in \mathcal{A}^{p+1}$. It is verified that:  $d\sigma\neq 0$ and $\delta\sigma = 0$.
\end{definition}

They verify: $d\mathcal{A}^{p-1} \subseteq Z^{p}/W^{p}$, $\delta\mathcal{A}^{p+1} \subseteq \hodgeop Z^{n-p}/W^{p}$ and $\mathcal{H}^{p} \subseteq W^{p}$, and they are linearly independent according the Proposition \ref{prop:linearindependence}.

\begin{proposition}\label{prop:nodeltaint}
	The following integrals are null.
	\begin{equation}
	\int_{z^{(p)}_{a}} d\alpha = 0 \qquad \int_{z^{(p)}_{a}} \delta\beta = 0
	\end{equation}
\end{proposition}
\begin{proof}
	The first is null because the cycles $z$ are closed sub-manifolds of $M$, therefore $\partial z=\emptyset$; based on the Stokes Theorem\cite{Gockeler:1989} it must be : $\int_{z^{(p)}_{a}} d\alpha = \int_{\partial z^{(p)}_{a}} \alpha = 0$. However, both first and second integrals are null based on the \poincare duality in Equation \eqref{eq:bottdefinition3}:
	\begin{equation}
	\int_{z^{(p)}_{a}} d\alpha = \eta \int_{M} d\alpha \wedge \gamma^{(n-p)}_{P(a)} = \eta \int_{M} \alpha \wedge \delta\gamma^{(n-p)}_{P(a)}=0
	\end{equation}
	
	\begin{equation}
	\int_{z^{(p)}_{a}} \delta\beta = \eta \int_{M} \delta\beta \wedge \gamma^{(n-p)}_{P(a)} = \eta \int_{M} \beta \wedge d\gamma^{(n-p)}_{P(a)}=0
	\end{equation}
\end{proof}

\begin{proposition}\label{prop:objects}
	The objects $p$-forms $\{d\alpha,\delta\beta, u_{a}\gamma^{(p)}_{a}\}$  are linearly independent on the operators $\{d,\delta, \int_{z^{(p)}_{a}}\}$. The only solution for:
	\begin{equation}\label{eq:linde}
	Ad\alpha +B\delta\beta+\sum C_{a}\gamma^{(p)}_{a}=0 
	\end{equation}
	is: $A=B=C_{a}=0$
\end{proposition}

\begin{proof}
	The objects $d\alpha \in d\mathcal{A}^{p-1}$, $\delta\beta \in \delta\mathcal{A}^{p+1}$ and $\gamma^{(p)}_{a}\in W^{p}$  are in linearly independent classes; this requires that expressions as Equation \eqref{eq:linde} are no possible with non null coefficients $A$, $B$ and $C_{a}$. 
\end{proof}

\begin{theorem}[Hodge decomposition]\label{th:simpledecomposicion1}
	If  $\phi$ is a $p$-form, 
	then there is solution for each one of the $\beta_p + 2$ linearly independent terms in the following decomposition:
	\begin{equation}\label{eq:i1}
	\phi = 
	d\alpha + \delta\beta +\sum_{a=1}^{\beta_{p}} u_{a}\gamma^{(p)}_{a} 
	\end{equation}
	where $\alpha$ is a $(p-1)$-form verifying $\delta\alpha=0$ and $\beta$ is a $(p+1)$-form verifying $d\beta=0$.
\end{theorem}

\begin{proof}
	Each one of the terms can be determined by applying suitable operator as shown in Table \ref{tb:operadores}. By using derivative:
	\begin{equation}
	d\phi = d\delta\beta = \Laplace\beta
	\end{equation}
	By applying coderivative:
	\begin{equation}
	\delta\phi = \delta d\alpha = \Laplace\alpha
	\end{equation}
	By applying integration in cycle $z^{(p)}_{b}$:
	\begin{equation}
	\int_{z^{(p)}_{b}}\phi = u_{b}
	\end{equation}
	There are solutions for $\alpha=G\circ \delta\phi$ and  $\beta=G\circ d\phi$. The forms $d\phi$ and $\delta\phi$ play the role of sources in the Laplace equations. The additional condition $\delta\alpha=0$  and $d\beta=0$ are gauge constraints which allow the solutions based on the Laplacian operator and its inverse the Green one. 
\end{proof}

The forms $\alpha$ and $\beta$ are not fully defined; there are some freedom degrees in they solutions. They admit some transformations as: $\alpha \rightarrow \alpha + d\xi$, where $\xi$ is a $(p-2)$-form, with $p-2 \ge 0$; also $\beta \rightarrow \beta + \delta\psi$, where $\psi$ is a $(p+2)$-form, with $p+2 \le n$.

\begin{table}[t]
	\centering
	\caption{Summary of operators and reverses. The additional gauge conditions: $\delta\alpha=0$ and $d\beta=0$ allow the use of a pair of mutually inverse operators: the Laplacian and Green.}\label{tb:operadores}
	\begin{tabular}{|c||c|c|c||c|}\hline
		Operator$\setminus$Object & $d\alpha$   & $\delta\beta$               & $u_{a}\gamma^{(p)}_{a}$ & Reverse\\\hline\hline
		$d$      & 0                   & $d\delta\beta= \Delta\beta$ & 0                       & $G\circ\Delta\beta$\\\hline 
		$\delta$      & $\delta d\alpha=\Delta\alpha$                   & 0 & 0                       & $G\circ\Delta\alpha$\\\hline  
		$\int_{z^{(p)}_{a}}$ & 0 & 0 & $u_{a}$ & $u_{a}\gamma^{(p)}_{a}$\\\hline
	\end{tabular}
\end{table}

\subsection{Properties of the auxiliary matrices}

\begin{proposition}\label{prop:simetria}
	The matrix $\mathbf{E}^{(p)}$ verifies that:
	\begin{equation}
	\mathbf{E}^{(p)} = (-1)^{(n-p)p}(\mathbf{E}^{(n-p)})^{T}
	\end{equation}
\end{proposition}
\begin{proof}
	It follows from the anti-commutative property of the wedge product, from it is verified that:
	\begin{equation}
	\int_{M}\gamma^{(p)}_{a}\wedge\gamma^{(n-p)}_{b} = (-1)^{(n-p)p}\int_{M}\gamma^{(n-p)}_{b}\wedge\gamma^{(p)}_{a}
	\end{equation}
	that implies:   $\varepsilon^{(p)}_{ab} = (-1)^{(n-p)p} \varepsilon^{(n-p)}_{ba}$.
\end{proof}

\begin{proposition}\label{prop:simpledecomposicion2}
	The Hodge dual, $\hodgeop\phi$, of the $p$-form  $\phi$ in Theorem \ref{th:simpledecomposicion1} is a $(n-p)$-form, so it admits the following decomposition:
	\begin{equation}\label{eq:i2}
	\hodgeop\phi = d\alpha' + \delta\beta' + \sum_{a=1}^{\beta_{n-p}} v_{a}\gamma^{(n-p)}_{a} \qquad v_{a} = \int_{z^{(n-p)}_{a}} \hodgeop\phi
	\end{equation}
	where $\alpha'$ is a $(n-p-1)$-form verifying $\delta\alpha'=0$ and $\beta'$ is a $(n-p+1)$-form verifying $d\beta'=0$.
\end{proposition}

\begin{proof}
	Each one of the terms can be determined by applying the same operators as in  Theorem \ref{th:simpledecomposicion1}.
	The forms $d\hodgeop\phi$ and $\delta\hodgeop\phi$ play the role of sources in the corresponding Laplace equations.
\end{proof}

\begin{remark}\label{rm:dual1}
	The $\hodgeop\phi$ form decomposed in Equation \eqref{eq:i2} can be alternatively obtained from Equation \eqref{eq:i1} by using Hodge duality operator. However, this requires some linear relationship between the sets $\hodgeop\gamma^{(p)}$ and $\gamma^{(n-p)}$. This is a clue that suggests some relationship between $\hodgeop H^{p}$ and $H^{n-p}$.
\end{remark}

\begin{proposition}\label{prop:dualgamma}
	The Hodge dual of the representative form, $\hodgeop\gamma^{(p)}_{a}$, admit the following expression:
	\begin{equation}\label{eq:dualvs}
	\hodgeop\gamma^{(p)}_{a} = \sum_{b=1}^{\beta_{n-p}} \tau^{(n-p)}_{ab}\gamma^{(n-p)}_{b}
	\end{equation}
	where the 
	matrix $\mathbf{T}^{(p)}=(\tau^{(p)}_{ab}) \in M(\beta_p,\mathbb{R})$ is defined as:
	\begin{equation}
	\tau^{(n-p)}_{ab}=\int_{z^{(n-p)}_{b}} \hodgeop\gamma^{(p)}_{a}
	\end{equation}
\end{proposition}

\begin{proof}
	It follows from the Proposition \ref{prop:simpledecomposicion2} by using $\phi = \gamma^{(p)}_{a}$:
	\begin{eqnarray}
	\hodgeop\gamma^{(p)}_{a} & = & d\alpha' + \delta\beta' + \sum_{b=1}^{\beta_{n-p}} v_{b}\gamma^{(n-p)}_{b} \\
	v_{b} & = & \int_{z^{(n-p)}_{b}}\hodgeop\gamma^{(p)}_{a} 
	\end{eqnarray}
	
	The Hodge dual $\hodgeop\gamma^{(p)}_{a}$ is a strong harmonic $(n-p)$-form, so it belongs to $W^{n-p}(M)$. Left side term and last right side term belong to the class $W^{n-p}(M)$. Due to the linear independence in Proposition \ref{prop:objects}, the terms $d\alpha'$ and $\delta\beta'$  must be null.  
\end{proof}

\begin{proposition}\label{prop:productmatrix}
	The product of matrices $\boldsymbol{T}^{(p)}$ and $\boldsymbol{T}^{(n-p)}$ verifies: $\boldsymbol{T}^{(n-p)}\times\boldsymbol{T}^{(p)}=(-1)^{D(p)}\, \boldsymbol{I}$
\end{proposition}

\begin{proof}
	From the Proposition \ref{prop:dualgamma} are obtained the following two expressions: 
	\begin{eqnarray}
	\hodgeop\gamma^{(p)}_{a} &=&  \sum_{b=1}^{\beta_{n-p}} \tau^{(n-p)}_{ab} \gamma^{(n-p)}_{b}\\
	\hodgeop\gamma^{(n-p)}_{b} &=&  \sum_{c=1}^{\beta_p} \tau^{(p)}_{bc} \gamma^{(p)}_{c}
	\end{eqnarray}
	
	By including the second Equation into the Hodge dual of the first one:
	\begin{eqnarray}
	\gamma^{(p)}_{a} &=&  \sum_{b=1}^{\beta_{n-p}}  \sum_{c=1}^{\beta_p}  (-1)^{D(p)}\tau^{(n-p)}_{ab} \tau^{(p)}_{bc} \gamma^{(p)}_{c}
	\end{eqnarray}
	it is conclude that:
	\begin{equation}
	\sum_{c=1}^{\beta_p}  \left[ \delta_{ac} - \sum_{b=1}^{\beta_{n-p}}  (-1)^{D(p)}\tau^{(n-p)}_{ab} \tau^{(p)}_{bc} \right]\gamma^{(p)}_{c} =0
	\end{equation}
	However, the representative forms are linearly independent. From this, it is obtained that must be:
	\begin{equation}
	\delta_{ac} - \sum_{b=1}^{\beta_p}  (-1)^{D(p)}\tau^{(n-p)}_{ab} \tau^{(p)}_{bc} =0 
	\end{equation}
\end{proof}

\begin{proposition}\label{prop:ecuacionET}
	The  matrix  $\boldsymbol{E}^{(p)}$ and $\boldsymbol{T}^{(n-p)}$  verify:
	\begin{equation}
	\mathbf{E}^{(p)}\times(\mathbf{T}^{(n-p)})^{T} = \boldsymbol{\Lambda}^{(p)}
	\end{equation}
	where $\boldsymbol{\Lambda}^{(p})\in M(\beta_p,\mathbb{R})$ is the matrix containing the elements $\lambda^{(p)}_{ab}=(\gamma^{(p)}_{a}, \gamma^{(p)}_{b})$.	
	The elements of  $\boldsymbol{T}^{(n-p)}$  are:	
	\begin{equation}
	\tau^{(n-p)}_{b,P(a)} = \frac{\lambda^{(p)}_{ab}}{\varepsilon^{(p)}_{a,P(a)}}  
	\end{equation}
\end{proposition}

\begin{proof}
	The bilinear integral $(\gamma^{(p)}_{a}, \gamma^{(p)}_{b})$ is defined as: 
	\begin{eqnarray}
	(\gamma^{(p)}_{a}, \gamma^{(p)}_{b}) & = & \int_{M}  \gamma^{(p)}_{a} \wedge \hodgeop\gamma^{(p)}_{b} \\ &=& \sum_{c=1}^{\beta_p} \tau^{(n-p)}_{bc} \int_{M}  \gamma^{(p)}_{a} \wedge \gamma^{(n-p)}_{c} \\ & = & \sum_{c=1}^{\beta_p} \tau^{(n-p)}_{bc}  \varepsilon^{(p)}_{ac}
	\end{eqnarray}
	The elements no null in matrix $\mathbf{E}^{(p)}$ are $\varepsilon^{(p)}_{a,P(a)}$, therefore it is obtained that: $(\gamma^{(p)}_{a}, \gamma^{(p)}_{b})  = \tau^{(n-p)}_{b,P(a)} \varepsilon^{(p)}_{a,P(a)} $. 
\end{proof}

\begin{proposition}\label{prop:bothintegrals}
	If $\phi$ is a $p$-form, it is verified that:
	\begin{eqnarray}
	\int_{z^{p}_{a}} \phi &=&  \sum_{b=1}^{\beta_{n-p}}(-1)^{D(p)} \tau^{(p)}_{ba}\int_{z^{n-p}_{b}} \hodgeop\phi\\
	\int_{z^{n-p}_{a}} \hodgeop\phi &=& \sum_{b=1}^{\beta_{p}} \tau^{(n-p)}_{ba}\int_{z^{p}_{b}} \phi
	\end{eqnarray}
\end{proposition}
\begin{proof}
	From the Theorem \ref{th:simpledecomposicion1} and Proposition \ref{prop:simpledecomposicion2} is obtained that:
	\begin{eqnarray}
	\phi & = & d\alpha + \delta\beta +\sum_{a=1}^{\beta_{p}} \left[ \int_{z^{(p)}_{a}}\phi \right]\gamma^{(p)}_{a}\label{eq:1}\\
	\hodgeop\phi &=& d\alpha' + \delta\beta' + \sum_{c=1}^{\beta_{n-p}} \left[ \int_{z^{(n-p)}_{c}}\hodgeop\phi\right]\gamma^{(n-p)}_{c}
	\end{eqnarray}
	
	By using Hodge dual to the last Equation and based on the result of Proposition \ref{prop:dualgamma}, it is obtained in successive steps:
	\begin{eqnarray}
	\phi &=& (-1)^{D(p)}\hodgeop d\alpha' + (-1)^{D(p)}\hodgeop\delta\beta' \\ && + \sum_{c=1}^{\beta_{n-p}} \left[ \int_{z^{(n-p)}_{c}}\hodgeop\phi\right](-1)^{D(p)}\hodgeop\gamma^{(n-p)}_{c}\\
	 &=& \delta\alpha'' + d\beta'' \\ &&+ \sum_{c=1}^{\beta_{n-p}} \left[ \int_{z^{(n-p)}_{c}}\hodgeop\phi\right](-1)^{D(p)} \left[  \sum_{a=1}^{\beta_p} \tau^{(p)}_{ca} \gamma^{(p)}_{a}\right]\\
	 &=& \delta\alpha'' + d\beta'' \\ &&+ \sum_{a=1}^{\beta_p} \left(
	(-1)^{D(p)} \tau^{(p)}_{ca} \left[ \int_{z^{(n-p)}_{c}}\hodgeop\phi\right]
	\right)\gamma^{(p)}_{a}\label{eq:2}
	\end{eqnarray}
	where $(-1)^{D(p)}\hodgeop d\alpha' = \delta\alpha'' = \hodgeop d\hodgeop \alpha''$, that is: $(-1)^{D(p)}\alpha'  = \hodgeop \alpha''$. Also: $ (-1)^{D(p)}\hodgeop\delta\beta' =  d \hodgeop\beta'= d\beta''$, that is:  $\hodgeop\beta'= \beta''$. 
	By comparing Equations (\ref{eq:1}) and (\ref{eq:2}) it is obtained the first of proposed expressions. The second expression is achieved based on the result of Proposition \ref{prop:productmatrix} applied to the first expression. 
\end{proof}

\section{Canonical Decomposition of the Norm}\label{sec:norm}

The Hodge Decomposition Theorem allows expressing any differential form by decomposing it in several canonical terms with specific formal properties. Similarly, its norm can be also decomposed in some canonical terms. 

\begin{proposition}\label{prop:norma1}
	Let $\phi$  be a $p$-forms with Hodge Decomposition as follows, also for its Hodge dual, $\hodgeop\phi$:
	\begin{eqnarray}
	\phi &=& d\alpha + \delta\beta + \sum_{a=1}^{\beta_{p}}u_{a}\gamma^{(p)}_{a} 
	\\
	\hodgeop\phi &=& d\alpha' + \delta\beta' + \sum_{a=1}^{\beta_{n-p}}v_{a}\gamma^{(n-p)}_{a} 
	\end{eqnarray}
	
	\begin{eqnarray}
	\delta\alpha=d\beta=\delta\alpha'=d\beta'=0 \\ u_{a}=\int_{z^{(p)}_{a}}\phi \qquad    v_{a}=\int_{z^{(n-p)}_{a}}\hodgeop\phi  
	\end{eqnarray}
	
	Its norm $(\phi,\phi)$ can be expressed as:
	\begin{equation}
	(\phi,\phi)= (d\alpha,d\alpha)+(\delta\beta,\delta\beta)+(\phi_{h},\phi_{h}) 
	\end{equation}
\end{proposition}

\begin{proof}
	If $A$ and $B$ are $p$-forms, they verifies:  $(A,B)=(B,A)$ and also verifies\cite{Flanders:2012}: $(dC,B)=(C,\delta B)$. From these properties it is follows that all the cross terms: $(d\xi,\delta\pi)$, $(d\xi,\gamma)$, $(\delta\xi,\gamma)$, $(d\alpha,\phi_{0})$ and $(\delta\beta,\phi_{0})$ are all null.. 
\end{proof}

\begin{proposition}\label{prop:norma4}
	It is verified that: $(d\alpha,d\alpha)= (\alpha,\delta\phi )$ and $(\delta\beta,\delta\beta)= (\beta,d\phi)$
\end{proposition}

\begin{proof}
	It follows  from: $(dC,B)=(C,\delta B)$, $d\phi = d\delta\beta$ and $\delta \phi = \delta d\alpha$. 
\end{proof}

\begin{proposition}\label{prop:norma2}
	It is verified that the cohomologous term, $(\phi_{h},\phi_{h})$, in Proposition \ref{prop:norma1}	can be decomposed as: 
	\begin{equation}
	(\phi_{h},\phi_{h}) = \sum_{a=1}^{\beta_{p}}  \varepsilon^{(p)}_{a,P(a)} u_{a} v_{P(a)}
	\end{equation}	
\end{proposition}

\begin{proof}
	From its definition:
	\begin{eqnarray}
	(\phi_{h},\phi_{h}) &=& \int_{M} \phi_{h}\wedge\hodgeop\phi_{h}\\ &=&\sum_{a=1}^{\beta_{p}}\sum_{a=1}^{\beta_{p}} u_{a} v_{b}  \int_{M} \gamma^{(p)}_{a}  \wedge  \gamma^{(n-p)}_{b} \\ &=& 
	\sum_{a=1}^{\beta_{p}}\sum_{a=1}^{\beta_{p}} u_{a} v_{b}  \varepsilon^{(p)}_{ab}=
	\sum_{a=1}^{\beta_{p}}  \varepsilon^{(p)}_{a,P(a)} u_{a} v_{P(a)}
	\end{eqnarray}	
\end{proof}

\begin{theorem}[Quantization of the Norm]\label{th:norm}
	The norm of a $p$-form  $\phi$ in  no simply connected manifold can be expressed from two integrals based on its sources, $d\phi$ and $\delta\phi$,  and from a discrete sum of values corresponding to every one of the $\beta_{p}$  cycles of the cohomology classes as follows:
	\begin{eqnarray}
	(\phi,\phi) &= &(\alpha,\delta\phi)+(\beta,d\phi) +
	\sum_{a=1}^{\beta_{p}}  \varepsilon^{(p)}_{a,P(a)} u_{a} v_{P(a)} 
	 \\
	 &&u_{a}=\int_{z_{a}^{(p)}}\phi \qquad    v_{a}=\int_{z^{p}_{a}}\hodgeop\phi 
	\end{eqnarray}
\end{theorem}

\section{Even-dimensional pseudo-Riemann manifolds}\label{sec:even}

Some special results can be obtained in even-dimensional pseudo-Riemann manifolds, where $n=2m$. In this case the \poincare duality for $p=m$ defines one endomorphism:  $H^{m} \rightarrow H^{m}$ where it is verified that for every representative form $\gamma^{(m)}_{a}$ exists one and only one $\gamma^{(m)}_{b}$ verifying:
\begin{equation}
\varepsilon^{(m)}_{ab}= \int_{M} \gamma^{(m)}_{a}\wedge\gamma^{(m)}_{b} \ne 0
\end{equation}

Therefore the Betti numbers $\beta_{m}$ are even in even $n=2m$ manifolds. 
Based on Proposition \ref{prop:simetria}, it must be:
\begin{equation}
\varepsilon^{(m)}_{ab}=  (-1)^{m^{2}}  \varepsilon^{(m)}_{ba} 
\end{equation}

For even $m$, the matrix $\mathbf{E}^{(m)}$ is symmetric, while for odd values it is antisymmetric. 

\begin{proposition}\label{prop:casoreal}
	In even dimensional manifolds, $n=2m$, the values of $\beta_{m}D(m)$ is even, therefore is verified that: $\boldsymbol{T}^{(m)}\in M(\beta_{m},\mathbb{R})$.
\end{proposition}

\begin{proof}
	The result of Proposition \ref{prop:productmatrix} applied to  $\boldsymbol{T}^{(m)}$ implies: $\boldsymbol{T}^{(m)}\times\boldsymbol{T}^{(m)}=(-1)^{D(m)}\, \boldsymbol{I}$. Thus, its determinant value is: $|\boldsymbol{T}^{(m)}|^{2}= (-1)^{\beta_{m}D(m)}=1$. It means that the determinant always verifies:  $|\boldsymbol{T}^{(m)}|= \pm 1$. With  $\boldsymbol{T}^{(m)}\in M(\beta_{m},\mathbb{R})$ the determinant is always real. In real domain, $\boldsymbol{T}^{(m)}$  belongs  to the GL($\beta_{m}$,$\mathbb{R}$) matrix Lie Group. 	
\end{proof}

\begin{remark}
	The $GL(n,\mathbb{R})$ Lie Group has two connected components different on the determinant sign\cite
	{Gorbatsevich:1997}. For the positive determinant, there is the subgroup $GL^{+}(n,\mathbb{R})$ while the negative determinant matrices do not form a Lie Group. It fails in the closure property. Additionally, the matrix must verify that its square must be the identity matrix with a positive sign.     
\end{remark}

\begin{proposition}\label{prop:restricciones}
	The matrices $\boldsymbol{E}^{(m)}$  and $\boldsymbol{\Lambda}^{(m)}$ verify the following constraint Equation:
	\begin{equation}
	\boldsymbol{\Lambda}^{(m)} \times \left( \boldsymbol{E}^{(m)}\right)^{-1} \times\boldsymbol{\Lambda}^{(m)} = (-1)^{D(m)} \boldsymbol{E}^{(m)}
	\end{equation} 
\end{proposition}

\begin{proof}
	It is a consequence of Propositions \ref{prop:productmatrix} and \ref{prop:ecuacionET}, with:
	\begin{equation}\label{eq:solucionT}
	\boldsymbol{T}^{(m)} = (\boldsymbol{\Lambda}^{(m)})^{T}\times ((\boldsymbol{E}^{(m)})^{-1})^{T}
	\end{equation}
\end{proof}

\begin{proposition}
	In  an even $n=2m$ pseudo-Riemann manifold,  $\phi\in\mathcal{A}^{m}$ can be expressed as:
	\begin{equation}\label{eq:decomp_m}
	\phi = d\alpha - \hodgeop  d\beta + \sum_{i=1}^{\beta_{m}}w_{a}\gamma^{(m)}_{a} \qquad \delta\alpha=\delta\beta=0
	\end{equation}
	
	where $\alpha\in \mathcal{A}^{m-1}$ and also $\beta\in \mathcal{A}^{m-1}$.
\end{proposition}

\begin{proof}
	Form the Hodge decomposition:
	\begin{equation}\label{eq:decomp_m2}
	\phi = d\alpha + \delta\theta + \sum_{i=1}^{\beta_{m}}w_{a}\gamma^{(m)}_{a} \qquad \delta\alpha=0 \quad d\theta=0
	\end{equation} 
	
	where $\theta \in \mathcal{A}^{m+1}$. From the definition of the operator $\delta$: as: $\delta\theta = (-1)^{C(m+1)}\hodgeop d\hodgeop\theta$, and rewriting: $(-1)^{C(m+1)}\hodgeop\theta  = -\beta$, therefore the gauge constraint $d\theta=0$ must be rewritten as: $\delta\beta=0$. 
	The minus sign in the $\beta$ term is arbitrary but highly convenient. 
\end{proof}

\begin{proposition}
	The $m$-form $\phi$ and its dual can be expressed as follows: 
	\begin{eqnarray}
	\phi          &=& d\alpha  - \hodgeop d\beta +  \phi_{h} \label{eq:forma1} \\
	\hodgeop\phi &=& \hodgeop d\alpha  +(-1)^{D(m)+1} d\beta + \hodgeop  \phi_{h} \label{eq:forma2}
	\end{eqnarray}
\end{proposition}

A more compact expression of the previous equations can be achieved by introducing the following matrix representation:
\begin{equation}\label{eq:compacta}
\left [  \begin{array}{r}
\phi \\
\hodgeop\phi
\end{array}\right]=
\left[\boldsymbol\sigma_{1} d
+ \boldsymbol\sigma_{2} (\hodgeop  d)\right] 
\left [  \begin{array}{c}
\alpha \\
\beta
\end{array}\right] + 
\sum_{a=1}^{\beta_{m}} 
\left [  \begin{array}{r}
u_{a} \\
v_{a}
\end{array}\right]
\gamma^{(m)}_{a}
\end{equation}

where: $u_{a}=\int_{z^{(m)}_{a}} \phi$ and $v_{a}=\int_{z^{(m)}_{a}}\hodgeop\phi$. The  $2\times2$ matrices $\boldsymbol\sigma_{1}$ and $\boldsymbol\sigma_{2}$ are:
\begin{equation}
\boldsymbol\sigma_{1} = \left ( \begin{array}{cc}
1 & 0 \\
0 & (-1)^{D(m)+1}
\end{array}\right ) \qquad \boldsymbol\sigma_{2} = \left ( \begin{array}{cc}
0 & -1 \\
1 & 0
\end{array}\right )
\end{equation}

For odd $D(m)$ both $\boldsymbol\sigma_{1}$ and $\boldsymbol\sigma_{2}$ define  the rotation matrix in the $SO(2)$ group: $\mathbf{R}(\xi)=\boldsymbol\sigma_{1}\cos\xi+\boldsymbol\sigma_{2}\sin\xi$.  The operator coderivative applied to this compact representation verifies:
\begin{equation}
\delta
\left [ 
\begin{array}{c}
\alpha \\
\beta
\end{array}
\right] = 0
\qquad
\delta
\left [  \begin{array}{r}
\phi \\
\hodgeop\phi
\end{array}\right] 
= \boldsymbol\sigma_{1} 
\Laplace
\left [  
\begin{array}{c}
\alpha \\
\beta
\end{array}
\right]
\end{equation}

\begin{proposition}\label{prop:norm_m}
	The norm of the $m$-form $\phi$ is decomposed as:
	\begin{eqnarray}
	(\phi,\phi) & = & (\alpha,\delta\phi)  + (-1)^{s} 
	(\beta, \delta\hodgeop\phi) \\ &&+
	\sum_{a=1}^{\beta_{m}}\varepsilon^{(m)}_{a,P(a)}u_{a}v_{P(a)} 
	\end{eqnarray}
\end{proposition}
\begin{proof}
	From the Hodge decomposition: $\phi = d\alpha + \delta\theta + \cdots$, such as its norm is: $(\phi,\phi)=(\alpha,\delta\phi) + (\theta,d\phi) + \cdots$. The second term has been modified as: $\delta\theta = - *d\beta$, that is equivalent to: $(-1)^{C(m+1)}\hodgeop\theta  = -\beta$, therefore:  $\theta  = (-1)^{C(m+1)+D(m+1)+1}\hodgeop\beta$. But by erasing the even terms:
	\begin{equation}
	C(m+1)+D(m+1)+1= m^{2}
	\end{equation}
	
	Hence, the second term of the norm must be changed as:
	\begin{eqnarray}
	(\theta,d\phi) &=& (-1)^{m^2}(\hodgeop\beta,d\phi)= (-1)^{m^2}(\beta,\hodgeop d\phi) \\ & = & (-1)^{m^2+D(m)}(\beta,\delta\hodgeop\phi) \\ &=& (-1)^{s}(\beta,\delta\hodgeop\phi)
	\end{eqnarray}
\end{proof}

\begin{proposition}\label{prop:relacionCompacta}
	It is verified the following compact relationship between topological integrals:
	\begin{equation}
	\left [  \begin{array}{c}
	u_{a} \\
	v_{a}
	\end{array}\right] =
	\sum_{b=1}^{\beta_{m}} 
	\tau^{(m)}_{ba}
	\left [  \begin{array}{cc}
	0 & (-1)^{D(m)}\\
	1 & 0
	\end{array}\right]
	\left [  \begin{array}{c}
	u_{b} \\
	v_{b}
	\end{array}\right]
	\end{equation}
\end{proposition}

\begin{proof}
	It is based on Proposition \ref{prop:bothintegrals}. The compatibility is just the Equation: $\boldsymbol{T}^{(m)}\times\boldsymbol{T}^{(m)}=(-1)^{D(m)}\, \boldsymbol{I}$.
\end{proof}

\begin{remark}\label{rmk:cuadratura}
	For odd $D(m)$ the matrix in right side on the previous Proposition is a rotation of $\pi/2$ in a plane.  The following complex representation can be introduced: $w = u+v\imath$, and the previous Equation becomes:
	\begin{equation}
	w_{a} =\imath
	\sum_{b=1}^{\beta_{m}} 
	\tau^{(m)}_{ba}
	w_{b} 
	\end{equation}
\end{remark}

\subsection{Solving the case $\beta_{m}=2$}\label{subsec:ejemplosbasicos}

The simplest solution is for $\beta_{m}=2$. The solutions of matrices will be designed by using simple coding criteria as S$\beta_{m}$.X.Y.. according to subcases that can be found in the analysis. 

Several procedures can be defined based in the relations of matrices $\boldsymbol{\Lambda}^{(m)})$, $\mathbf{E}^{(m)}$ and $\mathbf{T}^{(m)}$. The initial data can be the matrix $\mathbf{E}^{(m)}$ with the symmetries and \poincare duality relationship between cohomology classes. Using the result  of Proposition \ref{prop:restricciones} can be obtained the compatible values for $\boldsymbol{\Lambda}^{(m)})$ matrix. Other rather simple procedure to solve the matrices for low Betti numbers can be summarized as:
\begin{enumerate}
	\item  In te Equation \eqref{eq:solucionT}, to code the right hand side based on the defined matrix $\mathbf{E}^{(m)}$  and general $\boldsymbol{\Lambda}^{(m)})$.
	\item Based on  the symmetries of $\boldsymbol{\Lambda}^{(m)})$ and symmetries/antisymmetries of  $\mathbf{E}^{(m)}$, to identify the actual different terms in that side.
	\item To code the matrix $\mathbf{T}^{(m)}$ by including only the different elements.
	\item To solve the Equation in Proposition \eqref{prop:productmatrix}; some constraints can be required. 
\end{enumerate}

The $\mathbf{E}^{(m)}$ matrix can be symmetric or antisymmetric  depending on the value of $m$. The matrices are:
\begin{equation}
\mathbf{E}^{(m)} =
\left(
\begin{array}{cc}
0 & E_{12}\\
E_{21} & 0
\end{array}
\right) 
\qquad
\mathbf{\Lambda}^{(m)} =
\left(
\begin{array}{cc}
\lambda_{11} & \lambda_{12}\\
\lambda_{21} & \lambda_{22}
\end{array}
\right)
\end{equation}

The right side of Equation \eqref{eq:solucionT} is the following:
\begin{eqnarray}
(\boldsymbol{\Lambda}^{(m)})^{T} \times ((\mathbf{E}^{(m)})^{-1})^{T}  =
\left(
\begin{array}{cc}
\lambda_{21}/E_{21} & \lambda_{11}/E_{12}\\
\lambda_{22}/E_{21} & \lambda_{12}/E_{12}
\end{array}
\right) 
\end{eqnarray}

There are thee different terms from the symmetry relationships: $A=\lambda_{12}/E_{12}$, $B=\lambda_{11}/E_{12}$ and $C= \lambda_{22}/E_{12}$.  The matrix $\mathbf{T}^{(m)}$ is coded as follows:
\begin{equation}
\mathbf{T}^{(m)} =
\left(
\begin{array}{cc}
(-1)^{m^{2}}A & B\\
(-1)^{m^{2}} C & A
\end{array}
\right)
\end{equation}

It must ve verified that:  $\mathbf{T}^{(m)} \times  \mathbf{T}^{(m)} = (-1)^{D(m)} \mathbf{I}$, that generates the following equations:
\begin{eqnarray}
A^{2}+ (-1)^{m^{2}}BC &=& (-1)^{D(m)}\\
(-1)^{m^{2}}AB+ BA &=& 0\\
CA+ (-1)^{m^{2}} AC &=& 0\\
(-1)^{m^{2}}CB+ A^{2} &=& (-1)^{D(m)}
\end{eqnarray}

Several options with different solutions can be proposed generating a taxonomy of cases based on $m$ values.  

\begin{enumerate}
	\item \textbf{Even $m$}. The solution is obtained from $AB=0$, $AC=0$ and $A^{2}+ BC =(-1)^{s}$, that also implies other two possibilities; in both cases it is verified:
	\begin{equation}
	\mathbf{E}^{(m)} =
	\left(
	\begin{array}{cc}
	0 & E_{12}\\
	E_{12} & 0
	\end{array}
	\right) 
	\end{equation}
	
	\begin{enumerate}
		\item  [S2.1] $A=0$, that implies:  $BC =(-1)^{s}$. It means that: $\lambda_{12}=0$ and $\lambda_{11} \lambda_{22} =  (-1)^{s}(E_{12})^{2}$. The solution matrices are:
		\begin{eqnarray}
		\mathbf{\Lambda}^{(m)} &=&
		\left(
		\begin{array}{cc}
		\lambda_{11} & 0\\
		0 & \lambda_{22}
		\end{array}
		\right)
		\\
		\mathbf{T}^{(m)} &=&
		\left(
		\begin{array}{cc}
		0 & \lambda_{11}/E_{12}\\
		\lambda_{22}/E_{12} & 0
		\end{array}
		\right)
		\end{eqnarray}
		
		\item [S2.2] $B=C=0$, that implies:  $A^{2} =(-1)^{s}$, it is only possible, with real matrices, in metric with even $s$. It is verified: $\lambda_{11}=\lambda_{22}=0$ and: $(\lambda_{12})^{2} =  (-1)^{s}(E_{12})^{2}$. The solution matrices are:
		\begin{eqnarray}
		\mathbf{\Lambda}^{(m)} &=&
		\left(
		\begin{array}{cc}
		0 & \lambda_{12} \\
		\lambda_{12} & 0
		\end{array}
		\right)
		\\
		\mathbf{T}^{(m)} &=&
		\left(
		\begin{array}{cc}
		\lambda_{12}/E_{12} & 0\\
		0 & \lambda_{12}/E_{12} 
		\end{array}
		\right)
		\end{eqnarray}
		
	\end{enumerate}
	
	\item  [S2.3], \textbf{Odd $m$}. The following condition must be verified: $A^{2}-BC = (-1)^{s+1}$, that implies: $(\lambda_{12})^{2} - \lambda_{11} \lambda_{22} = (-1)^{s+1}(E_{12})^{2}$. The matrices are:
	\begin{eqnarray}
	\mathbf{E}^{(m)} &=&
	\left(
	\begin{array}{cc}
	0 & E_{12}\\
	-E_{12} & 0
	\end{array}
	\right) \\
	\mathbf{T}^{(m)} &=&
	\left(
	\begin{array}{cc}
	-\lambda_{12}/E_{12} & \lambda_{11}/E_{12}\\
	-\lambda_{22}/E_{12} & \lambda_{12}/E_{12} 
	\end{array}
	\right)
	\end{eqnarray}
\end{enumerate}

\begin{table}
	\centering
	\caption{Solution taxonomy for $\beta_{m}=2$ with their required  conditions.}\label{tb:soluciones}
	\begin{tabular}{|l l||c|c|c|c|c| }\hline
		Group & Solution & $m$  & $s$  & $|\mathbf{T}^{(m)}|$ & Description \\\hline\hline
		S2 &          &      &      &                      & One dual pair\\
		& S2.1   & even &  -   &      $(-1)^{s+1}$    &\\
		& S2.2   & even & even &       +1             &\\
		& S2.3   & odd  &   -  &     $(-1)^{s}$       &\\\hline
	\end{tabular}
\end{table}


\subsection{The Cohomology Forms and Matrices in the $2$-Torus}

The 2-Torus is a surface manifold $M \subset R^{3}$ with a  riemannian metric defined  as follows,  where: $u\in[0,2\pi]$ and $v\in[0,2\pi]$
\begin{equation}\label{eq:metrictorus}
ds^{2} = (R+r\cos v)^{2}du^{2}+ r^{2}dv^{2}
\end{equation}

The metric tensor $g_{ab}$ and its determinant $|g|$ are:
\begin{equation}
\mathbf{g}=\left ( 
\begin{array}{cc}
(R+r\cos v)^{2} & 0 \\
0 & r^{2}
\end{array}
\right) \qquad \sqrt{|g|}= (R+r\cos v)r
\end{equation}

\begin{figure}[t]
	\centering
	\includegraphics[width=0.3\textwidth]{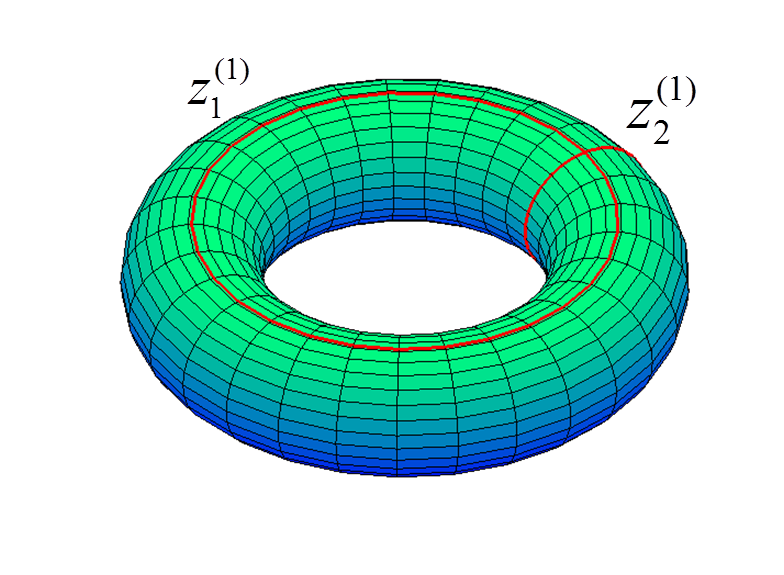}
	\caption{The $2$-Torus has two cohomologous $1$-cycles, $z^{(1)}_{1}$ and $z^{(1)}_{2}$, associated to cycles in coordinates $du$ and $dv$ respectively. Betti number is: $\beta_{1}=2$.} \label{fig:torosimple}	
\end{figure}

It is a even dimensional manifolds, $n=2$, $m=1$, pure Riemannian metric, $s=0$, $D(1)= 1$. With odd $m$ the matrix $\mathbf{E}^{(1)}$ must be antisymmetric. Table \ref{tb:figuras} shows Betti values.
The volume form is: $\Omega=\sqrt{|g|} \:du\wedge dv=(R+r\cos v)r\; du\wedge dv$. Its integral is the
manifold volume $\mathrm{vol}(M)$, in this case is the Torus surface area $A$:
\begin{eqnarray}
A &=& \mathrm{vol}(M)=\int_{M}\Omega= \int_{M} \sqrt{|g|}\;du\wedge dv \\
&=& \int_{M} (R+r\cos v)r \;du\wedge dv =Rr(2\pi)^{2}
\end{eqnarray}

\begin{table}[t]
	\centering
	\caption{Some remarkable surface manifolds $M\subset R^{3}$ with dimensionality $n=2$ and Riemannian metric with signature $s=0$. The manifold dimension is even, $n=2m$, and $m=1$ is odd, therefore it is a Solution S2.3.
		The Euler number is: $\chi =\beta_{0} -\beta_{1} +\beta_{2}$.}\label{tb:figuras}
	\begin{tabular}{|c|c|c|c|}\hline
		Surface & Euler: $\chi$ & Betti: $\beta_{0},\beta_{1},\beta_{2}$ & \poincare pairs\\\hline\hline
		Sphere & 2  & 1,0,1 & 0\\\hline       
		2-Torus  & 0  & 1,2,1 & 1\\\hline
		2-Torus g-2 & -2 & 1,4,1 & 2 \\\hline
		2-Torus g-3 & -4 & 1,6,1 & 3 \\\hline
	\end{tabular}
\end{table}

This is a manifold of finite volume, therefore the  unit $2$-form $\omega$ can be defined as: $\omega=\frac{1}{A}\Omega$. It can be expressed as:
\begin{eqnarray}
\omega &=& \frac{1}{A}\Omega=  \dfrac{1}{Rr(2\pi)^{2}}(R+r\cos v)r \;du\wedge dv \\ &=& \dfrac{du}{2\pi}\wedge \dfrac{dv}{2\pi} + d\left(\dfrac{r}{2\pi R}\sin v \;du \right)
\end{eqnarray}

The second term in the right side is an exact term $d\xi$, where $\xi= \frac{r}{2\pi R}\sin v \;du$. This allows to define an effective $\omega$ without the exact and vanishing term in the integral: 
\begin{equation}
\int_{M}\omega= \int_{M} \dfrac{du}{2\pi}\wedge \dfrac{dv}{2\pi} +\int_{M} d\xi = 1 +  \int_{\partial M} \xi =1
\end{equation}

because the manifold $M$ has not boundaries: $\partial M=\emptyset$. There are two linear independent cycles, two homology classes, therefore there are two cohomology $1$-forms, $\beta_{1}=2$, so they must be mutually \poincare duals. The two representative forms can be chosen as the related to $du$ and $dv$ respectively. These $1$-forms and its Hodge duals, that fit the unit $n$ form,  are:
\begin{alignat}{7}
&\gamma^{(1)}_{1} &=& \frac{du}{2\pi} &\qquad& \hodgeop\gamma^{(1)}_{1} &=&  \dfrac{dv}{2\pi}\\
&\gamma^{(1)}_{2} &=& \frac{dv}{2\pi} &\qquad& \hodgeop\gamma^{(1)}_{2} &=& -  \dfrac{du}{2\pi}
\end{alignat}
The norm integrals verify $(\gamma^{(1)}_{a},\gamma^{(1)}_{b})=\delta_{ab}$,  therefore $\mathbf{\Lambda}^{(1)}$ is diagonal,  the set is orthonormal that implies the Hodge and  \poincare dual must be proportional.
The elements of  matrix $\mathbf{T}^{(1)}$ are:
\begin{eqnarray}
\tau^{(1)}_{11} &=&  \int_{z^{1}_{1}} \hodgeop \gamma^{(1)}_{1} =  \int_{z^{1}_{1}} \dfrac{dv}{2\pi}= 0\\
\tau^{(1)}_{12} &=&  \int_{z^{1}_{2}} \hodgeop \gamma^{(1)}_{1} =  \int_{z^{1}_{2}} \dfrac{dv}{2\pi}= 1\\
\tau^{(1)}_{21} &=&  \int_{z^{1}_{1}} \hodgeop \gamma^{(1)}_{2} =  -\int_{z^{1}_{1}} \dfrac{du}{2\pi}= -1\\
\tau^{(1)}_{22} &=&  \int_{z^{1}_{2}} \hodgeop \gamma^{(1)}_{2} =  -\int_{z^{1}_{2}} \dfrac{du}{2\pi}= 0
\end{eqnarray}
The $\boldsymbol{T}^{(1)}$ matrix is:
\begin{equation}
\boldsymbol{T}^{(1)} = \left(  
\begin{array}{cc}
0 & 1\\
-1 & 0
\end{array}
\right) 
\qquad \boldsymbol{T}^{(1)} \times \boldsymbol{T}^{(1)} = -\boldsymbol{I}
\end{equation}
The values of the elements of matrix $\boldsymbol{E}^{(1)}$  are:
\begin{eqnarray}
\varepsilon^{(1)}_{11} &=&  \int_{M} \gamma^{(1)}_{1}\wedge\gamma^{(1)}_{1} =0 \\
\varepsilon^{(1)}_{12} &=&  \int_{M} \gamma^{(1)}_{1}\wedge\gamma^{(1)}_{2} =1 \\
\varepsilon^{(1)}_{21} &=&  \int_{M} \gamma^{(1)}_{2}\wedge\gamma^{(1)}_{1} =-1 \\
\varepsilon^{(1)}_{22} &=&  \int_{M} \gamma^{(1)}_{2}\wedge\gamma^{(1)}_{2} =0
\end{eqnarray}
The matrix $\boldsymbol{E}^{(1)}$ is:
\begin{equation}
\boldsymbol{E}^{(1)} = \left(  
\begin{array}{cc}
0 & 1\\
-1 & 0
\end{array}
\right) 
\qquad \boldsymbol{E}^{(1)}\times(\boldsymbol{T}^{(1)})^{T}  = \boldsymbol{I} 
\end{equation}

Proposition \ref{prop:restricciones} becomes verified.
This topology solution for the $2$-Torus is in the Group S2.3 in previous Subsection and Table \ref{tb:soluciones}; it has: $n=2m$, $m=1$ odd, $s=0$ even and $\beta_{1}=2$.  The constraint $(\lambda_{12})^{2} -
\lambda_{11} \lambda_{22} = (-1)^{s+1}(E_{12})^{2}$ is also verified.

\section{Application to Electromagnetism in Non-Simple Connected Manifolds}\label{sec:aplicacion}

An application is presented of the previous theoretical results for the Electromagnetic Field into a no simple connected manifold. Classic Physic, including Relativity and Electromagnetism, is mainly a theory in the simple connected manifold $\mathbb{R}^{4}$. If instead of the classic approach, the theory is developed in a $4$-dimensional no simple connected manifold, then some different results can arise due to the involved topology with its cohomologies.   

Following the Classical Theory\cite{Misner:1973,Hehl:2003,Landau:1973}, this field can be represented by using a $2$-form $F$. In this paper, it is used a general form, without the restriction of be exact. The elements of the $2$-form  $F$  include the $3$-dimensional vectors magnetic field, $\mathbf{B}$, and electric field, $\mathbf{E}/c$, where $c$ is the light speed in the vacuum; it is included to provide physical dimensional compatibility in the SI of units. It verify: $c^{2}=1/\mu_{0}\epsilon_{0}$. No material media is considered but currents and charge distributions are included. 

Let $(M,g,\gamma)$ be a $4$-dimensional compact, differentiable, oriented and connected pseudo-Riemann manifold $M$ with coordinates according 
MTW conventions\cite{Misner:1973}:  $(ct,x,y,z)$ having a metric ${g}$ locally reducible to a minkowskian  diagonal case with $s=1$, that a space-like metric: $\eta=\textrm{diag}(-1,1,1,1)$. 

Also, it must be considered $\gamma$ a set of representative cohomologous forms associated to closed sub-manifolds $z$. There are $\beta_{1}$ $1$-forms $\gamma^{(1)}$ and $\beta_{2}$ $2$-forms $\gamma^{(2)}$. This manifold is even dimensional, $n=4$, being $F$ a middle dimensional form, $m=2$. The signature $s$ is odd, and $m^{2}$ is even, therefore
the factor $D(2)$ is odd and  $D(1)=D(3)$ is even. The double Hodge duality is anti-symmetric for $2$-forms and symmetric for $1$-forms, the Betti number $\beta_{2}$  is even and the matrix $(\varepsilon^{(2)}_{ab})$ is symmetric. The cohomology representative forms are physically no dimensional as well as all the matrices generated from their. 

According the MTW conventions
the following  matrices that meet such criteria\cite
{Misner:1973}: 
\begin{equation}
(F)_{ab}= \left(
\begin{array}{cccc}
0   &   -E_{1}/c   &   -E_{2}/c & -E_{3}/c \\
E_{1}/c   &   0   &    B_{3} & -B_{2} \\
E_{2}/c   &  -B_{3}   &    0 & B_{1} \\
E_{c}/c &   B_{2} &  -B_{1} &  0 
\end{array}
\right)
\end{equation}
\begin{equation}
(\hodgeop F)_{ab}= \left(
\begin{array}{cccc}
0   &   B_{1}   &  B_{2} &  B_{3} \\
-B_{1}   &   0   &    E_{3}/c & -E_{2}/c \\
-B_{2}   &  -E_{3}/c   &    0 &  E_{1}/c \\
-B_{3} &   E_{2}/c &  -E_{1}/c &  0 
\end{array}
\right)
\end{equation}

They verify that:
\begin{eqnarray}\label{eq:desarrollo3}
\nabla_{b}F^{0b} &=&  \nabla\cdot \mathbf{E}/c \\  
\nabla_{b}F^{ab} &=&  -\frac{\partial E_{a}/c}{c\partial t} + (\nabla\times \mathbf{B})_{a} \\
\nabla_{b}(\hodgeop F)^{0b} &=&   - \nabla\cdot \mathbf{B} \\  
\nabla_{b}(\hodgeop F)^{ab} &=& \frac{\partial B_{a}}{c\partial t} + (\nabla\times \mathbf{E}/c)_{a} \label{eq:desarrollo4}
\end{eqnarray}

The contravariant components of $4$-dimensional currents vector $J^{a}$ are: $\{c\rho,J_{1},J_{2},J_{3} \}$ and the covariant ones $J_{a}$ are:  $\{-c\rho,J_{1},J_{2},J_{3}\}$. The two first Equations in the previous system are the Maxwell Equations that can be rewritten as: $\nabla_{b}F^{ab}=\mu_{0}J^{a}$, while the last two are:  $\nabla_{b}(\hodgeop F)^{ab}=0$.

However, the general $2$-form $F$ admits a Hodge decomposition, according to the results of Section \ref{sec:even}, but it is required to physically identify its sources. These are the continuous: $\delta F$, $\delta\hodgeop F$ and the discrete: $\int_{z} F$.

\subsection{The Physical meaning of Cohomology Integrals}
According to Classic Electromagnetism, the physical meaning of one integral as $\int_{z^{(2)}_{a}} \hodgeop F$ is one electric charge. However, the meaning of $\int_{z^{(2)}_{a}} F$ must be one magnetic charge, that is an unusual or exotic concept in Classical Theory, but have been largely considered in theoretical proposals of extension of the Classical Theory as Electromagnetic Duality\cite{Baez:1994, Jackson:1999, Olive:1996}. Both types previously defined are \emph{topological} charges, that is, they are associated to some integrals that can have no null values only in no simple connected manifolds. Thus, they depend mainly on the topological properties of the manifold. It must be remarked that this magnetic charge is not a free particle in the sense of Dirac magnetic monopoles~\cite{Dirac:1931}, instead this magnetic charge is only one integral, that is a global property of the manifold, one source of the field.    

The magnetic charges are theoretical concept, without experimental evidence. Thus, they have neither physical dimension nor units nor experimental values for the hypothetical charges.  There is some freedom degree in defining some of these properties, as the physical dimensionality of magnetic charges. In this paper, using these available freedom degrees, a dimensional relationship as $[q^{(E)}]=[q^{(M)}]$ is used. Remark that the chosen physical dimension implies that both are measured using the same units, but the numerical values of experimental charges, if they could exist, do not need to be the same that the electric ones. 

The Maxwell equations in Classic Electromagnetic Theory for the rationalized International System  of Units in differential form, generalization of Equations \eqref{eq:desarrollo3} to \eqref{eq:desarrollo4}, can be rewritten as:
\begin{equation}\label{eq:maxwelldual}
\delta F = \mu_{0} J^{(E)} \quad \delta \hodgeop F =  -{\mu_{0}} J^{(M)}
\end{equation}

where $J^{(E)}$ and $J^{(M)}$ are $1$-forms that are interpreted as the electric and magnetic relativistic currents. These $1$-forms include: $(c\rho,\mathbf{J})$, where its space components, $\mathbf{J}$, are the vector current and the time component is:  $c\rho$, where $\rho$ the volumetric density of charge. These equations are dimensionally consistent and the negative sign in the second equation is required to meet some compatibility criteria that must be explained later. From their definitions and the  property of double coderivative, $\delta\delta=0$, is verified that: $\delta J^{(E)} = \delta J^{(M)} = 0$. 

The equations can be rewritten  based on the double Hodge duality: $\hodgeop(\delta F)= d\hodgeop F$ and $\hodgeop(\delta \hodgeop F)= - d F$, that is symmetric for $(1,3)$-forms and is antisymmetric for $2$-forms. 
\begin{equation}\label{eq:maxwelldual2}
d\hodgeop F =  \mu_{0} (\hodgeop J^{(E)}) \quad dF =  {\mu_{0}} (\hodgeop J^{(M)})
\end{equation}

The vector expressions in $3$-vectors are the following
for the first Maxwell equation~\cite
{Baez:1994}: 
\begin{equation}\label{eq:maxwell1}
\nabla\cdot\mathbf{E}/c = \mu_{0} c\;\rho^{(E)} \qquad \nabla\times\mathbf{B}  - \frac{\partial\mathbf{E}/c}{c\partial t}= \mu_{0}\mathbf{J}^{(E)}
\end{equation}

while that second Maxwell equation~\cite
{Jackson:1999}  in $dF$ is:
\begin{equation}\label{eq:maxwell2}
-\nabla\cdot\mathbf{B} = -\mu_{0}c\rho^{(M)} \qquad \nabla\times\mathbf{E}/c +  \frac{\partial\mathbf{B}}{c\partial t} = -\mu_{0}\;\mathbf{J}^{(M)} 
\end{equation}

The integrals in a $2$-dimensional submanifold $\partial\Omega$   have the following meaning based on Gauss Theorem.
\begin{eqnarray}
\int_{\partial\Omega} \hodgeop F 
&= & \frac{1}{c} \int_{\partial\Omega} \mathbf{E} \cdot d\mathbf{s} + \cdots = \mu_{0}c\;Q^{(E)}
\end{eqnarray}

\begin{eqnarray}
\int_{\partial\Omega}  F''
&= & \int_{\partial\Omega} \mathbf{B} \cdot d\mathbf{s} + \cdots = \mu_{0}c\;Q^{(M)}
\end{eqnarray}

More compactly, the result can be achieved but using the Stokes Theorem: 
\begin{eqnarray}\label{eq:chargesOrd}
\int_{\partial\Omega} F &=& \int_{\Omega} dF = {\mu_{0}}\int_{\Omega} \hodgeop J^{(M)} = {\mu_{0}}c  \;Q^{(M)}\\ 
\int_{\partial\Omega} \hodgeop F &=& \int_{\Omega} d\hodgeop F = \mu_{0} \int_{\Omega} \hodgeop J^{(E)} = \mu_{0}c\;Q^{(E)}
\end{eqnarray} 

The criteria to choose the signs in Equation \eqref{eq:maxwelldual} is to provide positive value for the two previous integrals with meaning of charge. Based on this results, the cohomology integrals must be interpreted as:
\begin{equation}\label{eq:chargesTop}
\int_{z^{(2)}_{a}} F =  {\mu_{0}}c\; q^{(M)}_{a} \qquad \int_{z^{(2)}_{a}} \hodgeop F =  \mu_{0}c\; q^{(E)}_{a} 
\end{equation} 

where $q^{(M)}_{a} $ and  $q^{(E)}_{a} $ have the meaning  of magnetic and electric charges respectively associated to the $2$-cycles. However, in this context charges do not mean charged virtual~(point-like) particles.  These charges are generated by the topology properties and are specific to no simple connected manifolds. They are manifold properties and disappear in simple connected manifolds. However, the charges  $Q^{(M)}$ and   $Q^{(E)}$ are common in both simple and no simple connected manifold. Both concepts are formal different and coherently must be physically differentiated in a complete physic theory. Even though if $Q^{(M)}$ and $Q^{E}$ are null, there exists field.

\subsection{Double Potential and Quantization of the Norm}

In Classical Electromagnetic Theory $F$ is an exact $2$-form, $F=dA$, where the Potential $A$ is an $1$-form. However, it is not in a general case as is analysed in this paper. The main consequence is that there is not an equivalent to the $1$-form $A$. The Hodge decomposition of $F$ requires two $1$-forms as have been presented in Section \ref{sec:even}:
\begin{eqnarray}\label{eq:definicion}
F &=& dA^{(E)} - \hodgeop dA^{(M)} + {\mu_{0}}c \sum_{a=1}^{\beta_{2}} q^{(M)}_{a} \gamma^{(2)}_{a} \\
\hodgeop F &=& dA^{(M)} + \hodgeop dA^{(E)} + {\mu_{0}}c \sum_{a=1}^{\beta_{2}} q^{(E)}_{a} \gamma^{(2)}_{a}
\end{eqnarray}

The compacted Equation \eqref{eq:compacta} becomes: 
\begin{eqnarray}\label{eq:compactadaE}
\left [  \begin{array}{r}
F \\
\hodgeop F
\end{array}\right] &=&
\left[\boldsymbol\sigma_{1} d
+ \boldsymbol\sigma_{2} (\hodgeop  d)\right] \left [  \begin{array}{c}
A^{(E)} \\
A^{(M)}
\end{array}\right] \\ &&+ \mu_{0} c
\sum_{a=1}^{\beta_{2}} 
\left [  \begin{array}{r}
q^{(M)}_{a} \\
q^{(E)}_{a}
\end{array}\right]
\gamma^{(2)}_{a}
\end{eqnarray}

The continuous sources of $F$ are~\cite
{Misner:1973}: 
\begin{eqnarray}
\delta F &=& \delta dA^{(E)}  = \Laplace  A^{(E)} =  \mu_{0}J^{(E)} \\ \delta\hodgeop F &=& \delta dA^{(M)}  = \Laplace A^{(M)} =   - {\mu_{0}} J^{(M)} 
\end{eqnarray}

It must be remarked that conversely to Classical Electromagnetism, the sources $\delta F$ and $\delta\hodgeop F$ are not sufficient enough to determine the field $F$. It is required the additional discrete sources associated to the cohomology classes, that is, the topological charges. These charges $q^{(M)}_{a}$ and $q^{(E)}_{a}$ are not independent; they are related by the result of Proposition \ref{prop:bothintegrals} concerning both type of topology integrals: 
\begin{eqnarray}
q^{(M)}_{a}  &=&  - \sum_{b=1}^{\beta_{2}}  \tau^{(2)}_{ba} q^{(E)}_{b}\\
q^{(E)}_{a}  &=& \sum_{b=1}^{\beta_{2}} \tau^{(2)}_{ba} q^{(M)}_{b}
\end{eqnarray}

This Equations can be also rewritten as:
\begin{equation}\label{eq:compactadaCarga}
\left [  
\begin{array}{c}
q^{(M)}_{a} \\
q^{(E)}_{a}
\end{array}
\right]=
\sum_{b=1}^{\beta_{2}} 
\tau_{ba}
\left[  
\begin{array}{cc}
0 & -1 \\
1 & 0
\end{array}
\right] 
\left [  \begin{array}{c}
q^{(M)}_{b} \\
q^{(E)}_{b}
\end{array}
\right]
\end{equation}

The electric and magnetic charges of the manifold are not individual properties of each class of cohomology, but appear to be collective properties, so their values are such that they verify certain properties in the set of all classes.
If the compact we introduce the complex charge in the plane (M)-(E) as: $\mathbf{q}_{a}=q^{(M)}_{a}+\imath q^{(E)}_{a}$, then the right side is the rotation: $e^{\imath\pi/2}$. The previous Equation becomes in complex representation:
\begin{equation}\label{eq:compactadaCargaCompleja}
\mathbf{q}_{a} =
\imath\sum_{b=1}^{\beta_{2}} 
\tau_{ba} 
\mathbf{q}_{b} 
\end{equation}


Based on Proposition \ref{prop:norm_m}, the norm of $F$ can be decomposed in two continuous terms, that are generalization of the Classical Theory term $(A,J)$ and one discrete/quantized term with $\beta_{2}$ values:
\begin{eqnarray}\label{eq:energia}
(F,F) &=&  \mu_{0}(A^{(E)},J^{(E)}) +  {\mu_{0}} (A^{(M)},J^{(M)}) \\ &&+ (\mu_{0}c)^{2} \sum_{a=1}^{\beta_{2}}\varepsilon^{(2)}_{a,P(a)}  q^{(M)}_{a}  q^{(E)}_{P(a)}
\end{eqnarray}

\subsection{Quantization of the Action}

The electromagnetic action, $S$, in the Lagrangian formalism is:
\begin{equation}\label{eq:action8}
S = \frac{1}{c}\int_{M}\mathfrak{L} =  \frac{1}{c}\int_{M}\mathcal{ L}  \:d\Omega 
\end{equation}

where $\mathcal{ L}$ is the Lagrangian density and $\mathfrak{L}$ is its equivalent $4$-form. If the Maxwell equations must be generated from the variational principle over the action $S$, then the variations must be on each one of the $1$-forms  $A^{(E)}$ and $A^{(M)}$. In Classical Theory, only one of the two Maxwell equations set (the related to sources: electric currents and charges) is generated from the variational principle. The other is formal properties, they are a consequence of the mathematical framework: if $F$ is exact: $F=dA$, then $dF = ddA=0$, is a consequence of the mathematical representation framework; they do not come from any extremal principle.  However, if the two Maxwell equations are related to sources, as is the case of a general analysis of the field, both are required to be generated from the variational principle.   

The Lagrangian must be function of the field variables and its first derivatives, that is: $\mathcal{L}(A^{(E)},A^{(M)},F)$. The Lagrangian density $\mathcal{L}$  in tensor expression is the following, a generalization of the classic one: 
\begin{equation}\label{eq:action5}
\mathcal{L} =  - \frac{1}{4\mu_{0}} F^{ab}F_{ab} - A^{(E)}_{a}(J^{(E)})^{a} -  A^{(M)}_{a} (J^{(M)})^{a} 
\end{equation}

Based on the relationship between tensors,   $2$-form and  $1$-forms,
the Lagrangian  $4$-form $\mathfrak{L}$ must be:
\begin{equation}\label{eq:action2}
\mathfrak{L} =  - \frac{1}{\mu_{0}}
F\wedge\hodgeop F - A^{(E)}\wedge\hodgeop J^{(E)} -    A^{(M)}\wedge\hodgeop J^{(M)} 
\end{equation}

Therefore, the Action becomes: 
\begin{equation}
S = -\frac{1}{\mu_{0}c} (F,F) - \frac{1}{c}(A^{(E)}, J^{(E)}) - \frac{1}{c}(A^{(M)},J^{(M)})
\end{equation}

Based on Equation~\eqref{eq:energia}, the result is the following  Action with two continuous and one, $S_{d}$, discrete/quantized term. The amount of discrete terms is just the number of cohomology cycles, the Betti number $\beta_{2}$:
\begin{eqnarray}\label{eq:action3}
S &=&  - \frac{2}{c}(A^{(E)}, J^{(E)}) - \frac{2}{c}(A^{(M)},J^{(M)})\\ &&- {\mu_{0}c} \sum_{a=1}^{\beta_{2}}\varepsilon^{(2)}_{a,P(a)}  q^{(M)}_{a}  q^{(E)}_{P(a)} 
\end{eqnarray}

\subsection{Charges and Action with $\beta_{2}=2$}

It is verified that:  $n=2m$ even dimensional, $m=2$ is even and $s$ is odd, and $D(2)$ is odd; the matrix $\mathbf{E}^{(2)}$ is symmetric.  Hence, this is the solution S2.1 in Table \ref{tb:soluciones}.   

Thi solution with a pair of \poincare dual cohomologies has symmetric matrix $\mathbf{E}^{(2)}$ and verify the condition: $\lambda^{(2)}_{1}\lambda^{(2)}_{2}= - (\varepsilon^{(2)}_{12})^{2}$.
Thus, the norms $\lambda^{(2)}_{1}$ and $\lambda^{2}_{2}$ are in opposite sign, also it must be an orthogonal case with $\lambda^{(2)}_{12}=0$. The matrices are:
\begin{eqnarray}
\mathbf{E}^{(2)} &=&
\left(
\begin{array}{cc}
0 & \varepsilon^{(2)}_{12}\\
\varepsilon^{(2)}_{12} & 0
\end{array}
\right)
\\
\mathbf{T}^{(2)} &=&
\left(
\begin{array}{cc}
0 & \lambda^{(2)}_{1}/ \varepsilon^{(2)}_{12}\\
\lambda^{(2)}_{2}/ \varepsilon^{(2)}_{12} & 0 
\end{array}
\right)
\\
\mathbf{\Lambda}^{(2)} &=&
\left(
\begin{array}{cc}
\lambda^{(2)}_{1} & 0\\
0 & \lambda^{(2)}_{2}
\end{array}
\right) 
\end{eqnarray}

The relationship between magnetic and electric charges is in compact representation in Equation \eqref{eq:compactadaCarga} is:
\begin{equation}\label{eq:charges}
\left [  
\begin{array}{c}
q^{(M)}_{1} \\
q^{(E)}_{1}
\end{array}
\right]=
\frac{\lambda^{(2)}_{2}}{\varepsilon^{(2)}_{12}}
\left[  
\begin{array}{cc}
0 & -1 \\
1 & 0
\end{array}
\right] 
\left [  \begin{array}{c}
q^{(M)}_{2} \\
q^{(E)}_{2}
\end{array}
\right]
\end{equation}

that show that the charge relationship is composed of a scale factor and a rotation of $\pi/2$ in the plane $(M)-(E)$, such as the both charges are in quadrature. Due to $\lambda^{(2)}_{1}$ and $\lambda^{(2)}_{2}$ have different sign, there is a combination of net charge, or monopole, and differential charge, or dipole,  for both electric and magnetic components.  The monopole or net charge is the sum of charges: $m=q_{1}+q_{2}$, while the dipole moment is proportional to the differential charge: $d=q_{1}-q_{2}$. By using the following expression for each charges from the monopole and dipole factors:
\begin{equation}
q_{1} = \frac{1}{2} (m+d) \qquad q_{2} = \frac{1}{2} (m-d)
\end{equation}
from Equation \eqref{eq:charges},  
t is concluded the following mixture of magnetic and electric monopoles and dipoles:
\begin{equation}
(m^{(M)})^{2} - (d^{(M)})^{2} + (m^{(E)})^{2} - (d^{(E)})^{2} =0
\end{equation}
The discrete/quantized action, $S_{d}$,  expressed from both type of charges becomes:
\begin{eqnarray}\label{eq:action4}
S_{d} & = &  - {\mu_{0}c} \sum_{a=1}^{\beta_{2}}\varepsilon^{(2)}_{a,P(a)}  q^{(M)}_{a}  q^{(E)}_{P(a)} 
\\&=&  {\mu_{0}c} \left( \lambda^{(2)}_{1} (q^{(E)}_{1})^{2} + \lambda^{(2)}_{2} (q^{(E)}_{2} )^{2}  \right) 
\end{eqnarray}

The $\lambda$ factors have a different sign, therefore the Action has not defined sign. Even though there are charge configurations with null discrete Norm and Action. Although from a quantum theory viewpoint, have been argued~\cite{Baez:1994,Dirac:1931,Olive:1996}  that the product of electric and magnetic elementary charges is in the order of Plank constant; it must be considered that this depends on the used physical unit systems; the Gaussian one and $c=1$ are usually used for such comparative.  Equation \eqref{eq:action4} is consistent with such argument, in the relationship between the charge product and the action, but in SI physical units. Also, a topological magnitude is included: $\varepsilon^{(2)}_{12}$. 
The  physical magnitudes involved in the discrete Action are:
\begin{equation}\label{eq:magnitudes}
[h] = [\mu_{0}] [c] [e]^{2} [\lambda] \qquad  [\lambda] = \frac{[h]}{[\mu_{0}][c][e]^{2}} 
\end{equation}

The $\lambda$ magnitudes are physical no dimensional. If the values are in the order of the electron charge and one Plank quantum action, then  the $\lambda$ values are in the order  of the inverse of the fine structure constant $\alpha$ defined as:
\begin{equation}\label{eq:estructurafina}
\alpha = \frac{\mu_{0}c\:e^{2}}{2h}
\end{equation}

Thus, the value for $\lambda$ are in the order of the inverse of this constant:
\begin{equation}\label{eq:magnitude2s}
\lambda \simeq \frac{h}{\mu_{0}c[e]^{2}}  = \frac{1}{2\alpha}
\end{equation}

This result that for an hypothetical study in elementary particles and cohomologies, the following Equation can be an useful starting point, where actions becomes in the order of Plank constant and electric charges are in the order of the electron charge:
\begin{equation}\label{eq:fin1}
\int_{M} \gamma^{(2)}_{a}\wedge\hodgeop \gamma^{(2)}_{a} = \frac{\mathrm{cte}}{\alpha} 
\end{equation}

\section{Conclusions}\label{sec:conclusion}

Hodge decomposition has been presented as a practical application based on linear independent operators: the derivative, coderivative and cohomology integrals. They generate a  differential form decomposition in exact, dual exact and cohomology expansion, that have been proved to be complete.

That methodology is not norm dependent as most of presented in the scientific literature concerning Geometry and Topology. In Riemann manifolds the norm is positive definite and such property is extensively used in many proofs, but the more interesting for Physical applications are the pseudo-Riemann manifolds without that property. Thus, a methodology not founded on the norm signature is more useful to be applied in both Riemann and pseudo-Riemann manifolds.

A set of representative forms of the cohomologies has been used to define the \poincare duality. Some auxiliary matrices to characterize the relationship between Hodge and \poincare dualities of the representative forms have been proposed. 

The decomposition of differential forms in canonical terms requires a topology term obtained from the integral in the cohomology cycles. This decomposition has a counterpart in the decomposition of the norm in also canonical terms. The term concerning the cohomologies is a discrete/quantized finite sum whose number is just the Betti number.  The special case of even-dimensional manifolds has been analyzed including examples of representative forms and the auxiliary matrices. 

An application for the Electromagnetic Field in no simple connected manifold has been presented. A phenomenological interpretation of the cohomology integrals is needed. It is concluded that electric and magnetic charges must be required to correct interpretation of these cohomology integrals. That means that Electromagnetic Duality is necessary for the study of Electromagnetism in no simple connected manifolds.   However, the electric and magnetic charges of the manifold are neither free particles nor individual properties of each cohomology class, but appear to be collective properties of the manifold. It not possible to have electric monopoles without magnetic ones and vice versa. 

A significant result is the presence of one discrete/quantized term in the Field Norm and Action. The amount of such discrete values is the Betti number or \poincare pairs of cohomology classes. \emph{ A Classic Theory approach has been presented, but as result, the Action of Electromagnetic Field includes one quantized term}.  Even though if no continuous sources are present, there is a topology generated Electromagnetic Field whose Action is quantized.
The relation between magnetic monopoles and quantization of the action in the Electromagnetic Fields, widely suggested in previous theoretical studies in Quantum Physics, where the product of electric and magnetic charge is in the order of Plank constant, it is again confirmed but from a Classical Physics approach.

A well-known problem is Classical Electromagnetism happens when the integral of the field norm is extended to all the space-time; in this case, one infinite solution appears. It is an outcome of the point-like model of particles or the lack of a model for finite particles. In the proposed model of Electromagnetism coming from Topology, and in most of wormholes theories family, these integral in all the manifold do not generate infinite values, instead, there is a finite sum. The sources of these properties are cohomologies. 

\bibliographystyle{aipnum4-1}
\bibliography{bibliografiabasica}

\begin{thebibliography}{17}%
\makeatletter
\providecommand \@ifxundefined [1]{%
 \@ifx{#1\undefined}
}%
\providecommand \@ifnum [1]{%
 \ifnum #1\expandafter \@firstoftwo
 \else \expandafter \@secondoftwo
 \fi
}%
\providecommand \@ifx [1]{%
 \ifx #1\expandafter \@firstoftwo
 \else \expandafter \@secondoftwo
 \fi
}%
\providecommand \natexlab [1]{#1}%
\providecommand \enquote  [1]{``#1''}%
\providecommand \bibnamefont  [1]{#1}%
\providecommand \bibfnamefont [1]{#1}%
\providecommand \citenamefont [1]{#1}%
\providecommand \href@noop [0]{\@secondoftwo}%
\providecommand \href [0]{\begingroup \@sanitize@url \@href}%
\providecommand \@href[1]{\@@startlink{#1}\@@href}%
\providecommand \@@href[1]{\endgroup#1\@@endlink}%
\providecommand \@sanitize@url [0]{\catcode `\\12\catcode `\$12\catcode
  `\&12\catcode `\#12\catcode `\^12\catcode `\_12\catcode `\%12\relax}%
\providecommand \@@startlink[1]{}%
\providecommand \@@endlink[0]{}%
\providecommand \url  [0]{\begingroup\@sanitize@url \@url }%
\providecommand \@url [1]{\endgroup\@href {#1}{\urlprefix }}%
\providecommand \urlprefix  [0]{URL }%
\providecommand \Eprint [0]{\href }%
\providecommand \doibase [0]{http://dx.doi.org/}%
\providecommand \selectlanguage [0]{\@gobble}%
\providecommand \bibinfo  [0]{\@secondoftwo}%
\providecommand \bibfield  [0]{\@secondoftwo}%
\providecommand \translation [1]{[#1]}%
\providecommand \BibitemOpen [0]{}%
\providecommand \bibitemStop [0]{}%
\providecommand \bibitemNoStop [0]{.\EOS\space}%
\providecommand \EOS [0]{\spacefactor3000\relax}%
\providecommand \BibitemShut  [1]{\csname bibitem#1\endcsname}%
\let\auto@bib@innerbib\@empty
\bibitem [{\citenamefont {Flanders}(2012)}]{Flanders:2012}%
  \BibitemOpen
  \bibfield  {author} {\bibinfo {author} {\bibfnamefont {H.}~\bibnamefont
  {Flanders}},\ }\href@noop {} {\emph {\bibinfo {title} {Differential Forms
  with Applications to the Physical Sciences}}}\ (\bibinfo  {publisher} {Dover
  Publications},\ \bibinfo {year} {2012})\BibitemShut {NoStop}%
\bibitem [{\citenamefont {Morita}(2001)}]{Morita:2001}%
  \BibitemOpen
  \bibfield  {author} {\bibinfo {author} {\bibfnamefont {S.}~\bibnamefont
  {Morita}},\ }\href@noop {} {\emph {\bibinfo {title} {Geometry of Differential
  Forms}}}\ (\bibinfo  {publisher} {American Mathematical Society},\ \bibinfo
  {year} {2001})\BibitemShut {NoStop}%
\bibitem [{\citenamefont {G\"{o}ckeler}\ and\ \citenamefont
  {Sch{\"u}cker}(1989)}]{Gockeler:1989}%
  \BibitemOpen
  \bibfield  {author} {\bibinfo {author} {\bibfnamefont {M.}~\bibnamefont
  {G\"{o}ckeler}}\ and\ \bibinfo {author} {\bibfnamefont {T.}~\bibnamefont
  {Sch{\"u}cker}},\ }\href@noop {} {\emph {\bibinfo {title} {Differential
  Geometry, Gauge Theories, and Gravity}}}\ (\bibinfo  {publisher} {Cambridge
  University Press},\ \bibinfo {year} {1989})\BibitemShut {NoStop}%
\bibitem [{\citenamefont {Jost}(2011)}]{Jost:2011}%
  \BibitemOpen
  \bibfield  {author} {\bibinfo {author} {\bibfnamefont {J.}~\bibnamefont
  {Jost}},\ }\href@noop {} {\emph {\bibinfo {title} {Riemannian Geometry and
  Geometric Analysis}}}\ (\bibinfo  {publisher} {Springer},\ \bibinfo {year}
  {2011})\BibitemShut {NoStop}%
\bibitem [{\citenamefont {Hatcher}(2002)}]{Hatcher:2002}%
  \BibitemOpen
  \bibfield  {author} {\bibinfo {author} {\bibfnamefont {A.}~\bibnamefont
  {Hatcher}},\ }\href@noop {} {\emph {\bibinfo {title} {Algebraic Topology}}}\
  (\bibinfo  {publisher} {Cambridge University Press},\ \bibinfo {year}
  {2002})\BibitemShut {NoStop}%
\bibitem [{\citenamefont {Bott}\ and\ \citenamefont {Tu}(1982)}]{Bott:1982}%
  \BibitemOpen
  \bibfield  {author} {\bibinfo {author} {\bibfnamefont {R.}~\bibnamefont
  {Bott}}\ and\ \bibinfo {author} {\bibfnamefont {L.}~\bibnamefont {Tu}},\
  }\href@noop {} {\emph {\bibinfo {title} {Differential Forms in Algebraic
  Topology}}}\ (\bibinfo  {publisher} {Springer},\ \bibinfo {year}
  {1982})\BibitemShut {NoStop}%
\bibitem [{\citenamefont {{Misner}}\ and\ \citenamefont
  {{Wheeler}}(1957)}]{Misner:1957}%
  \BibitemOpen
  \bibfield  {author} {\bibinfo {author} {\bibfnamefont {C.~W.}\ \bibnamefont
  {{Misner}}}\ and\ \bibinfo {author} {\bibfnamefont {J.~A.}\ \bibnamefont
  {{Wheeler}}},\ }\href {\doibase 10.1016/0003-4916(57)90049-0} {\bibfield
  {journal} {\bibinfo  {journal} {Annals of Physics}\ }\textbf {\bibinfo
  {volume} {2}},\ \bibinfo {pages} {525} (\bibinfo {year} {1957})}\BibitemShut
  {NoStop}%
\bibitem [{\citenamefont {Bishop}\ and\ \citenamefont
  {Goldberg}(1968)}]{Bishop:1968}%
  \BibitemOpen
  \bibfield  {author} {\bibinfo {author} {\bibfnamefont {R.}~\bibnamefont
  {Bishop}}\ and\ \bibinfo {author} {\bibfnamefont {S.}~\bibnamefont
  {Goldberg}},\ }\href@noop {} {\emph {\bibinfo {title} {Tensor Analysis on
  Manifolds}}}\ (\bibinfo  {publisher} {Dover Publications},\ \bibinfo {year}
  {1968})\BibitemShut {NoStop}%
\bibitem [{\citenamefont {Friedlander}(1975)}]{Friedlander:1975}%
  \BibitemOpen
  \bibfield  {author} {\bibinfo {author} {\bibfnamefont {F.}~\bibnamefont
  {Friedlander}},\ }\href@noop {} {\emph {\bibinfo {title} {The Wave Equation
  on a Curved Space-time}}}\ (\bibinfo  {publisher} {Cambridge University
  Press},\ \bibinfo {year} {1975})\BibitemShut {NoStop}%
\bibitem [{\citenamefont {Gorbatsevich}, \citenamefont {Onishchik},\ and\
  \citenamefont {Vinberg}(1997)}]{Gorbatsevich:1997}%
  \BibitemOpen
  \bibfield  {author} {\bibinfo {author} {\bibfnamefont {V.}~\bibnamefont
  {Gorbatsevich}}, \bibinfo {author} {\bibfnamefont {A.~L.}\ \bibnamefont
  {Onishchik}}, \ and\ \bibinfo {author} {\bibfnamefont {E.~B.}\ \bibnamefont
  {Vinberg}},\ }\href@noop {} {\emph {\bibinfo {title} {Foundations of Lie
  Theory and Lie Transformations Groups}}}\ (\bibinfo  {publisher} {Springer},\
  \bibinfo {year} {1997})\BibitemShut {NoStop}%
\bibitem [{\citenamefont {Misner}, \citenamefont {Thorne},\ and\ \citenamefont
  {Wheeler}(1973)}]{Misner:1973}%
  \BibitemOpen
  \bibfield  {author} {\bibinfo {author} {\bibfnamefont {C.}~\bibnamefont
  {Misner}}, \bibinfo {author} {\bibfnamefont {K.}~\bibnamefont {Thorne}}, \
  and\ \bibinfo {author} {\bibfnamefont {J.}~\bibnamefont {Wheeler}},\
  }\href@noop {} {\emph {\bibinfo {title} {Gravitation}}}\ (\bibinfo
  {publisher} {W. H. Freeman},\ \bibinfo {year} {1973})\BibitemShut {NoStop}%
\bibitem [{\citenamefont {Hehl}\ and\ \citenamefont
  {Obukhov}(2003)}]{Hehl:2003}%
  \BibitemOpen
  \bibfield  {author} {\bibinfo {author} {\bibfnamefont {F.}~\bibnamefont
  {Hehl}}\ and\ \bibinfo {author} {\bibfnamefont {Y.}~\bibnamefont {Obukhov}},\
  }\href@noop {} {\emph {\bibinfo {title} {Foundations of Classical
  Electrodynamics: Charge, Flux, and Metric}}}\ (\bibinfo  {publisher}
  {Birkh{\"a}user Boston},\ \bibinfo {year} {2003})\BibitemShut {NoStop}%
\bibitem [{\citenamefont {Landau}\ and\ \citenamefont
  {Lifshitz}(1973)}]{Landau:1973}%
  \BibitemOpen
  \bibfield  {author} {\bibinfo {author} {\bibfnamefont {L.~D.}\ \bibnamefont
  {Landau}}\ and\ \bibinfo {author} {\bibfnamefont {E.}~\bibnamefont
  {Lifshitz}},\ }\href@noop {} {\emph {\bibinfo {title} {The Classical Theory
  of Fields}}},\ \bibinfo {edition} {fourth ed.}\ ed.\ (\bibinfo  {publisher}
  {Butterworth-Heinemann},\ \bibinfo {year} {1973})\BibitemShut {NoStop}%
\bibitem [{\citenamefont {Baez}\ and\ \citenamefont
  {Muniain}(1994)}]{Baez:1994}%
  \BibitemOpen
  \bibfield  {author} {\bibinfo {author} {\bibfnamefont {J.}~\bibnamefont
  {Baez}}\ and\ \bibinfo {author} {\bibfnamefont {J.}~\bibnamefont {Muniain}},\
  }\href@noop {} {\emph {\bibinfo {title} {Gauge Fields, Knots, and Gravity}}}\
  (\bibinfo  {publisher} {World Scientific},\ \bibinfo {year}
  {1994})\BibitemShut {NoStop}%
\bibitem [{\citenamefont {Jackson}(1999)}]{Jackson:1999}%
  \BibitemOpen
  \bibfield  {author} {\bibinfo {author} {\bibfnamefont {J.~D.}\ \bibnamefont
  {Jackson}},\ }\href@noop {} {\emph {\bibinfo {title} {Classical
  electrodynamics}}},\ \bibinfo {edition} {3rd}\ ed.\ (\bibinfo  {publisher}
  {Wiley},\ \bibinfo {address} {New York, {NY}},\ \bibinfo {year}
  {1999})\BibitemShut {NoStop}%
\bibitem [{\citenamefont {Olive}(1996)}]{Olive:1996}%
  \BibitemOpen
  \bibfield  {author} {\bibinfo {author} {\bibfnamefont {D.~I.}\ \bibnamefont
  {Olive}},\ }\href {\doibase http://dx.doi.org/10.1016/0920-5632(96)00002-3}
  {\bibfield  {journal} {\bibinfo  {journal} {Nuclear Physics B - Proceedings
  Supplements}\ }\textbf {\bibinfo {volume} {46}},\ \bibinfo {pages} {1 }
  (\bibinfo {year} {1996})}\BibitemShut {NoStop}%
\bibitem [{\citenamefont {Dirac}(1931)}]{Dirac:1931}%
  \BibitemOpen
  \bibfield  {author} {\bibinfo {author} {\bibfnamefont {P.~A.~M.}\
  \bibnamefont {Dirac}},\ }\href {\doibase 10.1098/rspa.1931.0130} {\bibfield
  {journal} {\bibinfo  {journal} {Proceedings of the Royal Society of London A:
  Mathematical, Physical and Engineering Sciences}\ }\textbf {\bibinfo {volume}
  {133}},\ \bibinfo {pages} {60} (\bibinfo {year} {1931})},\ \Eprint
  {http://arxiv.org/abs/http://rspa.royalsocietypublishing.org/content/133/821/60.full.pdf}
  {http://rspa.royalsocietypublishing.org/content/133/821/60.full.pdf}
  \BibitemShut {NoStop}%
\end{thebibliography}%

\end{document}